\documentclass{article}[11pt]

\usepackage[colorlinks=true, urlcolor=blue, linkcolor=blue, citecolor=blue]{hyperref}
\usepackage{makeidx}
\usepackage{latexsym}
\usepackage{amsfonts}
\usepackage{amssymb}
\usepackage{amsmath}
\usepackage{amstext}
\usepackage{amsthm}
\usepackage{mathrsfs}
\usepackage{mathabx}

\newtheorem{theorem}{Theorem}

\newtheorem{lemma}{Lemma}

\newtheorem{corollary}{Corollary}
\newtheorem{definition}{Definition}

\newtheorem{proposition}{Proposition}

\begin{document}

\title{Resolute and symmetric mechanisms \\ for two-sided matching problems}

\author{\textbf{Daniela Bubboloni} \\
{\small {Dipartimento di Matematica e Informatica U.Dini} }\\
\vspace{-6mm}\\
{\small {Universit\`{a} degli Studi di Firenze} }\\
\vspace{-6mm}\\
{\small {viale Morgagni 67/a, 50134 Firenze, Italy}}\\
\vspace{-6mm}\\
{\small {e-mail: daniela.bubboloni@unifi.it}}\\
\vspace{-6mm}\\
{\small tel: +39 055 2759667}\\
\vspace{-6mm}\\
{\small https://orcid.org/0000-0002-1639-9525}
\and
\textbf{Michele Gori}
\\
{\small {Dipartimento di Scienze per l'Economia e  l'Impresa} }\\
\vspace{-6mm}\\
{\small {Universit\`{a} degli Studi di Firenze} }\\
\vspace{-6mm}\\
{\small {via delle Pandette 9, 50127, Firenze, Italy}}\\
\vspace{-6mm}\\
{\small {e-mail: michele.gori@unifi.it}}\\
\vspace{-6mm}\\
{\small tel: +39 055 2759707}\\
\vspace{-6mm}\\
{\small https://orcid.org/0000-0003-3274-041X}
\and
\textbf{Claudia Meo}
\\
{\small {Dipartimento di Scienze Economiche e Statistiche} }\\
\vspace{-6mm}\\
{\small {Universit\`{a}} degli Studi di Napoli Federico II}\\
\vspace{-6mm}\\
{\small {via Cintia, 80126, Napoli, Italy}}\\
\vspace{-6mm}\\
{\small {e-mail:  claudia.meo@unina.it}}\\
\vspace{-6mm}\\
{\small tel: +39 081 675154 }\\
\vspace{-6mm}\\
{\small https://orcid.org/0000-0003-0870-3555}
}

\maketitle

\begin{abstract}
\noindent We focus on the one-to-one two-sided matching model with two disjoint sets of agents of equal size, where each agent in a set has preferences on the agents in the other set modeled by a linear order. A matching mechanism associates a set of matchings to each preference profile; resoluteness, that is the capability to select a unique matching, and stability are important properties for a matching mechanism. The two versions of the deferred acceptance algorithm are resolute and stable matching mechanisms but they are unfair since they strongly favor one side of the market. We introduce a property for matching mechanisms that relates to fairness; such property, called symmetry, captures different levels of fairness and generalizes existing notions. We provide several possibility and impossibility results mainly involving the most general notion of symmetry, known as gender fairness, resoluteness, stability, weak Pareto optimality and minimal optimality. In particular, we prove that: resolute, gender fair matching mechanisms exist if and only if each side of the market consists of an odd number of agents; there exists no resolute, gender fair, minimally optimal matching mechanism. Those results are obtained by employing algebraic methods based on group theory, an approach not yet explored in matching theory.
\end{abstract}
\vspace{4mm}


\section{Introduction}

In this paper we analyze issues about fairness for matching models  with two sides. In these models there are two disjoint sets of agents and each agent on one side has preferences over the agents on the other side. Agents on different sides have to be matched taking into account these preferences. A wide range of allocation problems from diverse fields can be analyzed within such a framework. Well-known examples include the college admissions market, the labour market for medical interns, auction markets and the living kidney donors market.

Our focus is, in particular, on the most basic example of matching models with two sides, the so called marriage model (Gale and Shapley, 1962), where agents on the two sides are interpreted as men and women. For simplicity, we adopt the terminology related to Gale and Shapley's metaphor along the paper. We denote the two disjoint sets of agents by $W$ and $M$ and refer to their elements as women and men, respectively; the set $W \cup M$ is denoted by $I$ and its elements are called individuals. We assume that the sizes of the two sets $W$ and $M$ are equal.\footnote{Most applications involve sets of agents having different sizes; however, in some circumstances assuming the same size is not a restriction. Consider, for instance, a university which offers to its near-graduate students the opportunity to experience a short-term work experience in selected public and private companies. Assume further that there are $n$ internship programs available and that each program is interested to select exactly one student. If the university preliminarily shortlists $n$ students among the ones applying, we face a situation where $n$ students have to be assigned to $n$ internship programs. A good assignment should consider both the preferences of each student over the companies who offer such internships and the preferences of each company over the students.} Moreover, the preference relation of each woman is modeled by a linear order on the set $M$; analogously, the preference relation of each man is modeled by a linear order on the set $W$. Since the two sets $W$ and $M$ are fixed, a marriage model is described by a profile of linear orders, one for each individual. Given a marriage model, one is interested to determine a matching, that is,  a bijective function from $I$ to $I$ which associates each woman with a man, each man with a woman, and has the additional property of being equal to its inverse.

Stability is a highly desirable property a matching may meet. A matching is stable if there is no blocking pair, namely, a woman and a man who both prefer each other to their current partners in the matching. In their seminal paper, Gale and Shapley (1962) proved that every marriage model admits a stable matching and also described an algorithm that finds such a matching. The algorithm involves more rounds. In the first round, every woman makes a proposal to the man she prefers most; every man who receives proposals from different women chooses his most preferred woman and gets temporarily matched with her, while all the other women who proposed to him are rejected.  In each subsequent round, each unmatched woman makes a proposal to her most-preferred man to whom she has not yet proposed (regardless of whether that man is already matched), and each man who receives some proposals gets matched to the woman he prefers most among the ones who proposed to him and the woman he has been already matched to, if any. In particular, if he has a provisional partner and he prefers another woman to her, he rejects the provisional partner who becomes unmatched again. This process is repeated until all women have been matched to a partner. The role of the two groups of individuals can be reversed with men proposing matches and women deciding whether to accept or reject each proposal. In general, the stable matchings produced by the two versions of the algorithm are different. Gale and Shapley also proved that the stable matching generated by the algorithm when women propose is optimal for all the women, in the sense that it associates each woman with the best partner she can have among all the stable matchings. For this reason, it is called woman-optimal stable matching. However, it is the worst stable matching for the other side of the market: indeed, it associates each man with the worst partner he can get among all the stable matchings.\footnote{This property is a consequence of a more general result due to Knuth (1976). See, also, Roth and Sotomayor (1990), Theorem 2.13 and Corollary 2.14.} Symmetrically, the man-optimal stable matching is the stable matching that results from the algorithm when the men propose: it is the best stable matching for men and the worst stable matching for women.

The purpose of matching theory is, however, more general than determining a set of matchings for a specific marriage model. In fact, its objective is to determine a method able to select a sensible set of matchings for any conceivable marriage model, namely a correspondence that associates a set of matchings with each preference profile. Such a correspondence is called a matching mechanism. Hence, a matching mechanism operates as a centralized clearinghouse that collects the preferences  of all market participants and provides a set of matchings; whether those matchings can be determined by using efficient algorithms is an important question for practical applications that has received considerable attention. Several properties may be imposed on a matching mechanism.
First of all, a matching mechanism is required to be non-empty valued. Especially when used for practical applications, a matching mechanism should also be resolute, that is, it should be able to select a unique matching for each preference profile. Stability is a further desirable property for matching mechanisms, particularly from a practical point of view: a matching  mechanism is stable if it selects a set of stable matchings for each preference profile. Among stable matching  mechanisms, the two mechanisms $GS_w$ and $GS_m$, which associate  with each preference profile the woman-optimal stable matching and the man-optimal stable matching respectively, are very popular and have been implemented in practice, for instance, in the National Resident Matching Program (NRMP) to match doctors to residency programs or fellowships and also in some school districts to match students to high schools (Roth, 1984, Abdulkadiro\u{g}lu and S\"{o}nmez, 2003). In terms of applications, there is evidence that stable matching mechanisms perform better than unstable ones in those situations where agents voluntarily accept the proposed matching; in fact, in many cases, matching  mechanisms that are not stable had largely failed and been abandoned while the stable ones succeeded and survived (see, for example, Roth, 1991). However, when matchings are imposed to agents, other properties involving equity and fairness appear more compelling than stability. Indeed, a matching mechanism is expected to be non-discriminatory towards agents, that is, some form of equity should be satisfied. The two mechanisms $GS_w$ and $GS_m$ clearly fail such objective because they favor one side of the market (women and men, respectively) over the other. This has motivated the design of alternative matching mechanisms able to select stable matchings other than the woman-optimal and the man-optimal stable matchings in order to reduce the conflict of interests between women and men, and to treat the two parts more fairly. We can mention, among others, the Sex-Equal matching mechanism (Gusfield and Irving, 1989 and  Romero-Medina, 2001) and the Egalitarian Stable matching mechanism (Gusfield and Irving, 1989).\footnote{For the definitions of these two mechanisms see Section \ref{ResProb}.}   All these mechanisms share a common approach: first, a measure is introduced to quantify the fairness of a matching; then, such measure is optimized (usually, minimized due to its interpretation) on the set of stable matchings.  A recap of several of these matching mechanisms can be found in Cooper and Manlove (2020).

In this paper, we take a different approach to fairness based on the idea that individuals should be treated independently on their identities: a mechanism is fair if a change in the identities of the individuals results into the same change in the output.\footnote{Our interpretation of fairness is somehow close to the concept of procedural fairness analyzed in Klaus and Klijn (2006) and based on probabilistic considerations: a (random) mechanism is procedurally fair whenever each agent has the same probability to move at a certain point in the procedure that determines the final probability distribution.}
Given the bilateral nature of the model, equity may concern either just agents belonging to the same group (women or men) or the totality of individuals. Indeed, in many situations, it is natural to require that no woman should gain a systematic advantage over the other women on the basis of her identity and, symmetrically, no man should gain a systematic advantage on the other men on the basis of his identity. In other words, an equal-treatment of individuals should be guaranteed within each group. This property is called anonymity and has been considered in the literature on two-sided matchings (see, for instance, Sasaki and Toda, 1992,  and Ha{\l}aburda, 2010). On the other hand, in other circumstances, it could be natural to impose a form of equity across the two sets of individuals by requiring that neither women nor men should have an advantage on the other group solely on the basis of their gender. A fairness notion of the latter kind has been firstly introduced in Masarani and Gokturk (1989) under the name of gender indifference: roughly speaking, it states that if we swap agents belonging to different sets in a specific way (that is, we rename women as men and men as women according to a given fixed rule), this should result in a corresponding change of names in the matchings selected by the matching mechanism.\footnote{Masarani and Gokturk (1989) only consider resolute matching mechanisms (or, matching algorithms, as they call them). Moreover, they account only for a specific change in the identities of individuals, namely, if $W=\{1,\ldots, n\}$ and $M=\{n+1,\ldots, 2n\}$, then woman 1 becomes man $n+1$ and vice versa, woman 2 becomes man $n+2$ and vice versa, and so on. Endriss (2020) and Root and Bade (2023) also analyse gender indifference, with the latter authors referring to this property as weak gender-neutrality.} A more general version of gender indifference has been provided by \"{O}zkal-Sanver (2004) under the name of gender fairness. Gender fairness, which implies anonymity\footnote{See Proposition 3.1 in \"{O}zkal-Sanver (2004). See also Section \ref{SymMM} in this paper.}, takes into account more general changes in the identities of individuals: women are renamed as men and men are renamed as women in any possible fashion.

In this paper we introduce a general notion of fairness for matching mechanisms that we call symmetry. Anonymity, gender indifference and gender fairness can be seen as special instances of this notion. In order to provide an informal definition of symmetric mechanisms, we first need some preliminary notions about how to model changes both in the identities and, subsequently, in the preferences of the individuals. 

First of all, we use permutations over the set $I$ of individuals, namely bijective functions from $I$ to $I$,  in order to model changes in the individual identities. In particular, we are interested in changes in the identities that account for the presence of two different types of individuals in the market. Hence, we focus on permutations $\varphi$ over $I$ that keep the partition $\{W,M\}$ of $I$ fixed, that is, $\{\varphi(W), \varphi(M)\}= \{W,M\}$. We denote by $G^*$ the set of such permutations. A permutation $\varphi$ in $G^*$ may refer to a change in individuals' identities either within the sets $W$ and $M$ or across the sets $W$ and $M$. In the former case, $\varphi$ is such that $\varphi(W)=W$ and $\varphi(M)=M$; we denote by $G$ the set of such permutations. In the latter case, $\varphi$ is such that $\varphi(W)=M$ and $\varphi(M)=W$.

Second, we account for the fact that every permutation in $G^*$ also affects preference profiles, namely, the complete list of individual preference relations. Suppose, for instance, that $W=\{1,2,3\}$ and $M=\{4,5,6\}$ and let $\varphi\in G^*$ be the permutation that switches $1$ and $4$, $2$ and $5$, $3$ and $6$. If woman 1's preferences on $M$ are $[5,6,4]$\footnote{Throughout the paper we represent linear orders with ordered lists with obvious meaning: for instance, the writing $[5,6,4]$ representing woman 1's preferences on $M$ means that woman 1 prefers man 5  to man 6, and man 6 to man 4.} in a given preference profile, when individual identities are changed according to $\varphi$, we get that man 4's preferences on $W$ are $[2,3,1]$ in the preference profile modified by $\varphi$. Given a preference profile $p$ and a permutation $\varphi\in G^*$, we denote by  $p^\varphi$ the preference profile obtained from $p$ by changing individuals' identities according to $\varphi$. 

The notion of symmetry that we introduce is the following. Consider a subset $U$ of $G^*$. We say that a matching mechanism $F$ is $U$-symmetric if, for every preference profile $p$ and every $\varphi\in U$, the set of matchings associated with $p^{\varphi}$ by $F$ corresponds to the set of the matchings associated with $p$ by $F$, each of them modified according to $\varphi$. Roughly speaking, $F$ is $U$-symmetric whenever a change in individuals' names in the input according to a permutation in $U$ entails the same change of names in the matchings of the output.

The definition of $U$-symmetry is very general and allows to model diverse levels of fairness depending on which subset $U$ of $G^*$ is considered. Moreover, the notion of $U$-symmetry is able to recover the notions of fairness that we have listed before by choosing suitable sets $U$. In fact, gender fairness coincides with the notion of $(G^*\setminus G)$-symmetry; anonymity is equivalent to $G$-symmetry; gender indifference coincides with $\{\varphi\}$-symmetry, where $\varphi$ is a matching, namely a permutation such that $\varphi(W)=M$ and $\varphi(M)=W$ and that coincides with its inverse. It is also worth mentioning that if $U\subseteq V\subseteq G^*$, then $V$-symmetry implies $U$-symmetry and so $G^*$-symmetry is the strongest possible level of symmetry a matching mechanism may meet.

By analyzing the matching mechanisms most frequently used in the literature, a clear contrast between resoluteness and symmetry emerges. Indeed, the matching mechanisms $GS_w$ and $GS_m$ are resolute, $G$-symmetric (i.e. anonymous) but, as one may expect, fail to be $G^*$-symmetric. On the other hand, the mechanism that associates with any preference profile the set of the stable matchings is $G^*$-symmetric but far from being resolute. Analogously, the Sex-Equal and the Egalitarian Stable matching mechanisms, introduced by Gusfield and Irving (1989) to reduce the conflict of interests between the two parts by selecting suitable stable matchings, are both $G^*$-symmetric but not resolute.\footnote{See Propositions \ref{STsym} and \ref{SEsym} in Appendix \ref{GT2}.}

On the basis of those observations, in this paper we undertake the following research project: under the assumption that $W$ and $M$ have the same size, studying the existence of matching mechanisms that are resolute, satisfy some form of symmetry and, possibly, fulfill some optimality properties like, for instance, stability.

The following two strong impossibility and possibility theorems, which are the main findings in our paper, are about the strongest possible level of symmetry and do not involve any optimality condition: 
\begin{itemize}
  \item[(A)] If the size of $W$ is even, then there exists no resolute and $G^*$-symmetric matching mechanism (Theorem \ref{n2}).
  \item[(B)] If the size of $W$ is odd, then there exists a resolute and $G^*$-symmetric matching mechanism (Theorem \ref{main-mathc}).
\end{itemize}
Of course, given the possibility result (B), a natural question is whether the existence of a resolute and $G^*$-symmetric matching mechanism continues to be guaranteed when one also asks for some optimality condition. We provide a negative answer; indeed, we prove that resoluteness and $G^*$-symmetry together cannot coexist either with a very weak optimality condition that we introduce in our paper and call minimal optimality.\footnote{A matching $\mu$ is minimally optimal if there is at least an individual who does not get her/his worst choice under $\mu$ (Definition \ref{DefOptCond}). A matching mechanism is minimally optimal if it only selects minimally optimal matchings (Definition \ref{mm-prop}).}
\begin{itemize}
\item[(C)] There exists no resolute, $G^*$-symmetric and minimally optimal matching mechanism (Theorem \ref{no-sym-mo}).
\end{itemize}
Since minimal optimality is weaker than both Pareto optimality and weakly Pareto optimality\footnote{A matching $\mu$ is Pareto optimal if there is no other matching $\mu'$ that makes at least one individual better off and no individuals worse off; weakly Pareto optimal if there is no other matching $\mu'$ that makes all the individual better off (Definition \ref{DefOptCond}). A matching mechanism is Pareto optimal [weakly Pareto optimal] if it only selects Pareto optimal [weakly Pareto optimal] matchings (Definition \ref{mm-prop}).} which, in turn, are weaker than stability (see Proposition \ref{implications}), the previous result implies that there are no resolute, $G^*$-symmetric and stable [Pareto optimal; weakly Pareto optimal] matching mechanisms.

From (C), it turns out that resolutness and $G^*$-symmetry are basically incompatible with any type of optimality condition. Moreover, for every $U\subseteq G$, that is, for every set of permutations that exchange individuals' identities only within the two groups $W$ and $M$, resoluteness and $U$-symmetry are consistent with stability. Indeed, as already observed, $GS_w$ and $GS_m$ are resolute, stable and $G$-symmetric, and then they are $U$-symmetric, as well.
We investigate then whether some optimality conditions may be consistent with resoluteness and a minimal level of symmetry among those that allow for an exchange of individuals' identities between the two groups $W$ and $M$, that is, $\{\varphi\}$-symmetry, where $\varphi\in G^*\setminus G$ is a fixed matching.
The results we found are summarized below:
\begin{itemize}
\item[(D)] There exists a resolute, $\{\varphi\}$-symmetric and weakly Pareto optimal matching mechanism (Theorem \ref{existence3}).
\item[(E)] If the size of $W$ is 2, then there exists a resolute, $\{\varphi\}$-symmetric and stable matching mechanism (Theorem \ref{existence2}).
\item[(F)] If the size of $W$ is at least 3, then there exists no resolute, $\{\varphi\}$-symmetric and stable matching mechanism (Theorem \ref{endriss}).
\end{itemize}
Thus, while (D) shows that resoluteness and $\{\varphi\}$-symmetry are consistent with weak Pareto optimality, they are not consistent with stability, unless the number of women and men is two.

We remark that some impossibility results concerning stability, resoluteness and a sort of fairness have been provided in the literature about matching mechanisms. Masarani and Gokturk (1989) basically show that resoluteness, $G^*$-symmetry, stability and a further axiom called maximin optimality cannot coexist. Endriss (2020) proves instead (F) by using tools from logic and computer science. In particular, he first shows that stability and gender indifference can be encoded in a specific formal language; then, he proves that an impossibility result involving those properties can be reduced to an impossibility result where the sizes of $W$ and $M$ is three, that this latter case can be fully automated using Satisfiability (SAT) solving technology, and that the automated proof can be made human-readable. We present a proof of (F) which differs from the one provided by Endriss.

In fact, our paper is based on a strongly innovative methodology. Specifically, all our main results are proved by exploiting notions and techniques from group theory. While an algebraic approach based on symmetric groups and groups actions has been already successfully used in social choice literature (E\u gecio\u glu, 2009, E\u gecio\u glu and Giritligil, 2013,  Bubboloni and Gori, 2014, 2015, 2016, 2021, 2023, Do\u gan and Giritligil, 2015,  Bartholdi et al., 2021,  Kivinen, 2023, 2024), at the best of our knowledge, this is the first time that such an approach has been employed in matching theory to provide general possibility and impossibility results.

Concepts and theorems from group theory turn out to be powerful tools in situations where permutations play a important role. Since matchings are themselves special permutations of the set $I$ and the various levels of fairness a matching mechanism may meet can be described through suitable sets of permutations of the set $I$, the application of group theory within matching theory emerges naturally. The use of group theory allows for new methods to prove known results and helps in finding new results. Indeed, as already said, in our paper we provide a new proof of (F) and, more significantly, we also provide original results that likely cannot be proved without a deep understanding of suitable permutation groups. For example, the proof of (B) is based, among other things, on a detailed analysis of the centralizers of the semiregular subgroups of $G^*$ (Theorem \ref{main-semireg2}), and it seems difficult to figure out an alternative proof. Furthermore, group theory improves our understanding of the concept of symmetry. In particular, it is possible to prove that a matching mechanism is symmetric with respect to a set $U$ of permutations of $I$ if and only if it is symmetric with respect to the group generated by $U$ (Proposition \ref{iff}). That fact has significant consequences: for instance, using non-obvious results from group theory, it can be shown that the combination of anonymity and gender indifference is equivalent to $G^*$-symmetry, and that gender fairness is also equivalent to $G^*$-symmetry. Moreover, knowing that only subgroups matter, we get that the possible levels of symmetry of a matching mechanism are as many as the subgroups of $G^*$.

The paper is organized as follows. Section 2 contains some basic preliminaries about groups, permutations and relations needed for the model. In Section 3 we illustrate the setup and the key notion of symmetric mechanism. In Section \ref{reserach-problem} we state our research problem and we present the main results of the paper. A generalization to the model with outside option is presented in Section 5. Finally, Section 6 concludes. The appendices collect several notions and results about matchings, matching mechanisms and group theory and most of the proofs. Appendix \ref{semi-reg}, in particular, is crucial as it presents technical results from group theory which play a central role in the proofs of the two main theorems of our paper.

\section{Basic preliminaries}
Throughout the paper, $\mathbb{N}$ denotes the set of positive integers.
Given $k\in \mathbb{N}$, we set $ \ldbrack k \rdbrack\coloneq \{n\in \mathbb{N}: n\leq k\}$.

\subsection{Groups and permutations}

A finite group is a finite set endowed with an operation that is associative, admits a neutral element  and has the property that each element admits an inverse. In a generic group the neutral element is denoted by $1$. However, for specific groups the notation can be specialized.\footnote{We refer to Milne (2021) for a reader-friendly but wide introduction to group theory, which turns out to be substantially exhaustive for the purposes of our paper.} 

Let $X$ be a finite set. $\mathrm{Sym}(X)$ denotes the group of the bijective functions from $X$ to itself with the operation defined, for every $\varphi_1,\varphi_2\in \mathrm{Sym}(X)$, by the right-to-left composition $\varphi_1\varphi_2\in \mathrm{Sym}(X)$.\footnote{One of the reasons why we refer to Milne (2021) instead of to more classical references, such as Robinson (1996), is that Milne adopts the right-to-left composition while  usually the left-to-right is used.} The neutral element of $\mathrm{Sym}(X)$ is the identity function on $X$, denoted by $id_X$. 
$\mathrm{Sym}(X)$ is called the symmetric group on $X$ and its elements are called permutations of $X$.

Given $k\in \mathbb{N}$ with $2\leq k\leq |X|$, we say that $\varphi\in \mathrm{Sym}(X)$ is a $k$-cycle if
there exist distinct $x_1,\dots, x_k\in X$ such that $\varphi(x_i)=x_{i+1}$ for all $i\in \ldbrack k-1 \rdbrack$, $\varphi(x_k)=x_{1}$, and $\varphi(x)=x$ for all $x\in X\setminus \{x_1,\dots, x_k\}$. In such a case, $\varphi$ is denoted by $(x_1\dots x_k)$, the set $\{x_1,\dots, x_k\}$ is called the support of $\varphi$ and $k$ is called the length of $\varphi$.
For instance, if $X=\ldbrack 5 \rdbrack$, $\varphi=(134)$ denotes the element of $\mathrm{Sym}(X)$ such that $\varphi(1)=3$, $\varphi(2)=2$, $\varphi(3)=4$, $\varphi(4)=1$, $\varphi(5)=5$, and is a $3$-cycle. We say that a permutation of $X$ is a cycle if it is a $k$-cycle for some $k\in \mathbb{N}$ with $2\leq k\leq |X|$. A result on permutations (see Proposition 4.26 in Milne, 2021) assures that any permutation of $X$ different from $id_X$ can be expressed, in a unique way up to reordering, as the product of commuting cycles whose supports are pairwise disjoint. Such a product is called  the cycle decomposition of the permutation. For instance, if $X=\ldbrack 6 \rdbrack$ and $\varphi\in \mathrm{Sym}(X)$ is defined by $\varphi(1)=3$, $\varphi(2)=4$, $\varphi(3)=1$, $\varphi(4)=6$, $\varphi(5)=5$, $\varphi(6)=2$, we have that $\varphi=(13)(246)=(246)(13)$. In what follows, we always write permutations different from the identity using their cycle decomposition.

\subsection{Relations}\label{sec-relation}

Let $X$ be a nonempty and finite set. A relation on $X$ is a subset of $X^2$. The set of the relations on $X$ is denoted by $\mathbf{R}(X)$. 

Let $R\in \mathbf{R}(X)$. Given $x,y\in X$, we usually write $x\succeq_R y$ instead of $(x,y)\in R$ and  $x\succ_R y$ instead of $(x,y)\in R$ and $(y,x)\notin R$. We say that $R$ is complete if, for every $x,y\in X$, $x\succeq_R y$ or $y\succeq_R x$;
antisymmetric if, for every $x,y\in X$, $x\succeq_R y$ and $y\succeq_R x$ imply $x=y$;
transitive if, for every $x,y,z\in X$, $x\succeq_Ry$ and $y\succeq_R z$ imply $x\succeq_R z$;
a linear order if it is complete, transitive and antisymmetric.
The set of linear orders on $X$ is denoted by $\mathbf{L}(X)$.

Let $R\in \mathbf{L}(X)$. Given $x\in X$, we set $\mathrm{Rank}_R(x)=|\{y\in X: y\succeq_{R} x\}|$. If $|X|=n$, we represent $R$ by an ordered list $[x_1,\ldots, x_n]$, where $X=\{x_1,\ldots, x_n\}$ and, for every $i,j\in \ldbrack n \rdbrack$,  $x_i\succeq_R x_j$ if and only if $i\le j$. Such a representation completely encodes the relation $R$. Note also that if $R=[x_1,\ldots, x_n]$, then, for every $i\in \ldbrack n \rdbrack$, $\mathrm{Rank}_R(x_i)=i$. Thus, for every $x\in X$, the number $\mathrm{Rank}_R(x)$ represents the position of $x$ in the ordered list representing $R$.

Let $R\in \mathbf{R}(X)$, $Y$ be a finite set with $X\subseteq Y$, and $\varphi\in\mathrm{Sym}(Y)$. We denote by
$\varphi R$ the relation on $ \varphi(X)$ defined by
\begin{equation}\label{prodottorel}
\varphi R \coloneq\left\{(x,y)\in \varphi(X)^2: (\varphi^{-1}(x),\varphi^{-1}(y))\in R\right\}.
\end{equation}
Note that, for every $x,y\in X$, $(x,y)\in R$ if and only if $(\varphi(x),\varphi(y))\in\varphi R$.
It can be easily checked that if $\varphi_1,\varphi_2\in\mathrm{Sym}(Y)$, then $(\varphi_1\varphi_2)R=\varphi_1(\varphi_2 R)$. Hence, the writing $\varphi_1\varphi_2 R$ is not ambiguous.

The following lemma will be fundamental to set up our model in the next section.

\begin{lemma}\label{phiRR}
Let $X$ and $Y$ be nonempty and finite sets with $X\subseteq Y$, $R\in \mathbf{L}(X)$ and $\varphi\in\mathrm{Sym}(Y)$. 
Then $\varphi R\in \mathbf{L}(\varphi(X))$; $\varphi R=R$ if and only if, for every $x\in X$, $\varphi(x)=x$; if 
$|X|=m$ and $R=[x_1,\ldots, x_m]$, where $\{x_1,\ldots, x_m\}=X$, then $\varphi R=[\varphi(x_1),\ldots, \varphi(x_m)]$.
\end{lemma}

\begin{proof}
Assume $|X|=m$. Since $R\in \mathbf{L}(X)$, we know that $R$ can be represented by an ordered list $[x_1,\ldots, x_n]$, where $\{x_1,\ldots, x_m\}=X$. 
Of course, $\varphi(X)=\{\varphi(x_1),\ldots, \varphi(x_m)\}$. Consider $i,j\in \ldbrack m \rdbrack$. We know that $x_i\succeq_R x_j$ if and only if $i\le j$. By the definition of $\varphi R$, we deduce that $\varphi(x_i)\succeq_{\varphi R} \varphi(x_j)$ if and only if $x_i\succeq_{R} x_j$. Thus, we have that $\varphi(x_i)\succeq_{\varphi R} \varphi(x_j)$ if and only if $i\le j$. Then, we conclude that $\varphi R\in \mathbf{L}(\varphi(X))$ and $\varphi R=[\varphi (x_1),\ldots, \varphi (x_m)]$. In particular, if, for every $x\in X$, $\varphi(x)=x$, then we get $\varphi R=R$. Assume now, conversely, that $\varphi R=R$. Then, the relations $\varphi R$ and $R$ must be defined on the same set and so $\varphi(X)=X$. Moreover, they must have the same representation as ordered lists. Thus, we have that $[x_1,\ldots, x_m]=[\varphi(x_1),\ldots, \varphi(x_m)]$, which implies that $\varphi(x_i)=x_i$ for all $i\in  \ldbrack m \rdbrack$, that is, $\varphi(x)=x$ for all $x\in X$.
\end{proof}

\section{Setup of the model}\label{setup}

\subsection{Preference profiles}\label{model}

Let $W$ and $M$ be two finite disjoint sets with $|W|=|M|=n\ge 2$ and let $I=W\cup M$.
We refer to the sets $W$, $M$ and $I$ as the set of women, the set of men and the set of individuals, respectively. For simplicity, we assume $W=\{1,\ldots,n\}$ and $M=\{n+1,\ldots,2n\}$. In order to simplify the reading, we adopt the following rule throughout the paper (with few exceptions): the letter $x$ is used to denote the elements of $W$, the letter $y$  is used to denote the elements of $M$ and the letter $z$  is used to denote the elements of $I$, that is, individuals for whom it is not important to specify or it is not known  whether they belong to $W$ or $M$.

Each woman has preferences on the set $M$ and each man has preferences on the set $W$. We assume that the preferences of women are represented by linear orders on $M$ and the preferences of men are represented by linear orders on $W$. A preference profile $p$ is a function from $I$ to $\mathbf{L}(W)\cup \mathbf{L}(M)$ such that, for every $x\in W$,  $p(x)\in \mathbf{L}(M)$ and, for every $y\in M$, $p(y)\in \mathbf{L}(W)$. We denote by $\mathcal{P}$ the set of preference profiles. We represent a preference profile $p\in\mathcal{P}$ by the table
\[
\begin{array}{|ccc||ccc|}
\hline
1&\ldots&n&n+1&\ldots&2n\\
\hline
\hline
p(1)&\ldots&p(n)&p(n+1)&\ldots&p(2n)\\
\hline
\end{array},
\]
where in the first row there are the names of individuals and in the second row the corresponding preferences.
For instance, if $W=\{1,2,3\}$ and $M=\{4,5,6\}$, the table
\begin{equation}\label{tablep}
\begin{array}{|ccc||ccc|}
\hline
1&2&3&4&5&6\\
\hline
\hline
4&4&6&2&3&3\\
5&6&5&1&1&2\\
6&5&4&3&2&1\\
\hline
\end{array}
\end{equation}
represents the preference profile $p$ such that
\[
p(1)=[4,5,6]\in \mathbf{L}(M), \; p(2)=[4,6,5]\in \mathbf{L}(M), \; p(3)=[6,5,4]\in \mathbf{L}(M),
\]
\[
p(4)=[2,1,3]\in \mathbf{L}(W), \; p(5)=[3,1,2]\in \mathbf{L}(W), \; p(6)=[3,2,1]\in \mathbf{L}(W).
\]
Note that, in the first row of the table representing a preference profile, the elements of $I$ are written in increasing order. However, a reordering of the columns is always allowed. For example, the table in \eqref{tablep} and the tables
\[
\begin{array}{|ccc||ccc|}
\hline
3&2&1&4&6&5\\
\hline
\hline
6&4&4&2&3&3\\
5&6&5&1&2&1\\
4&5&6&3&1&2\\
\hline
\end{array}\quad\quad 
\begin{array}{|ccc||ccc|}
\hline
4&5&6&2&1&3\\
\hline
\hline
2&3&3&4&4&6\\
1&2&1&6&5&5\\
3&1&2&5&6&4\\
\hline
\end{array}
\]
all represent the same preference profile.

\subsection{Matchings} \label{Matchings}

A matching is a permutation $\mu\in \mathrm{Sym}(I)$ such that, for every $x\in W$, $\mu(x)\in M$; for every $y\in M$, $\mu(y)\in W$; for every $z\in I$, $\mu(\mu(z))=z$. The set of matchings is denoted by $\mathcal{M}$.
Since $y=\mu (x)$ implies $\mu(y)=\mu(\mu (x))=x$, $\mu$ decomposes into the product of $n$ disjoint $2$-cycles, each of them having the form $(z\  \mu(z))$. Hence, the cycle decomposition of $\mu$ clearly  illustrates what are the couples formed under the matching $\mu$.

Consider, for instance,  $W=\{1,2,3\}$ and $M=\{4,5,6\}$. Then we have that 
\[
\mathcal{M}=\Big\{(14)(25)(36),(14)(26)(35),(16)(25)(34),(15)(24)(36),(15)(26)(34),(16)(24)(35)\Big\}.
\]
Thus, for example, the matching $(14)(25)(36)$ matches woman $1$ with man $4$, woman $2$ with man $5$ and woman $3$ with man $6$.

Let us introduce some optimality conditions a matching may meet. The first three are well known.
The last one, to the best of our knowledge, is new.
\begin{definition}\label{DefOptCond} {\rm Let $p\in \mathcal{P}$ and $\mu\in\mathcal{M}$. We say that $\mu$ is:
\begin{itemize}
\item {\it stable} for $p$ if there is no pair $(x,y)\in W\times M$ such that $y\succ_{p(x)} \mu(x)$ and $x\succ_{p(y)} \mu(y)$;
\item {\it Pareto optimal} for $p$ if there is no matching $\mu'\in\mathcal{M}$ such that, for every $z \in I$, $\mu'(z) \succeq_{p(z)} \mu(z)$  and, for some  $z^* \in I$, $\mu'(z^*) \succ_{p(z^*)} \mu(z^*)$;
\item {\it weakly Pareto optimal} for $p$ if there is no matching $\mu'\in\mathcal{M}$ such that, for every $z \in I$, $\mu'(z) \succ_{p(z)} \mu(z)$;
\item {\it minimally optimal} for $p$ if there exists $z\in I$ such that $\mathrm{Rank}_{p(z)}(\mu(z))<n$.
\end{itemize}}
\end{definition}
Thus, a matching is stable if there is no pair of individuals formed by a woman and a man such that both prefer each other to the partners assigned to them by the matching; such pairs are called, in the literature, \lq\lq blocking pairs\rq\rq. A matching is not Pareto optimal for $p$ if there exists another matching making some individual better off without making any other individual worse off;  not weakly Pareto optimal for $p$ if there exists another matching making every individual better off; not minimally optimal for $p$ if all the individuals are matched with their least preferred choice. 
It is worth mentioning that, given a preference profile $p$, there exists at most one matching that is not minimally optimal for $p$ (see Proposition \ref{soloMO} in Appendix \ref{OptCond}). The following proposition describes the relation among the properties in Definition \ref{DefOptCond}. Its proof is in Appendix \ref{OptCond}.

\begin{proposition}\label{implications}
Let $p\in \mathcal{P}$ and $\mu\in\mathcal{M}$. Then the following facts hold true:
\begin{itemize}
\item[$(i)$]if $\mu$ is stable for $p$, then $\mu$ is Pareto optimal for $p$;
\item[$(ii)$]if $\mu$ is Pareto optimal for $p$, then $\mu$ is weakly Pareto optimal for $p$;
\item[$(iii)$]if $\mu$ is weakly Pareto optimal for $p$, then $\mu$ is minimally optimal for $p$.
\end{itemize}
\end{proposition}

Gale and Shapley (1962) proved that, for every $p\in \mathcal{P}$, a stable matching exists and provided an algorithm to find such a matching. In this algorithm, each side is assigned a specific role: according to their preferences, individuals on one side make proposals to individuals on the other side who are allowed to accept or refuse the proposals. Hence, the algorithm has two distinct versions, each producing a stable matching. Those matchings, in general, differ from each other, but both possess a special property. Indeed, Gale and Shapley proved that the matching resulting from the algorithm where women make proposals associates each woman with the best possible partner she can have within stable matchings, and each man with the worst possible partner he can have within stable matchings. Symmetrically, the algorithm where men make proposals associates each man with the best possible partner he can have within stable matchings, and each woman with the worst possible partner she can have within stable matchings. For such reasons, in the literature, these two matchings are frequently referred to as the woman-optimal stable matching and the man-optimal stable matching, respectively.

\subsection{Matching mechanisms}\label{secMM}

A matching mechanism is a correspondence from $\mathcal{P}$ to $\mathcal{M}$.
Thus, a matching mechanism is a procedure that allows to select a set of matchings for any given preference profile.

We denote by $TO$ the matching mechanism that associates with any $p\in \mathcal{P}$ the whole set $\mathcal{M}$;
by $ST$ the matching mechanism that associates with any $p\in \mathcal{P}$ the set of matchings that are stable for $p$;
by $GS_w$ [$GS_m$] the matching mechanism that associates with any $p\in \mathcal{P}$ the set having as unique element the woman-optimal [man-optimal] stable matching for $p$;
by $GS$ the matching mechanism that associates with any $p\in \mathcal{P}$ the set whose elements are the woman-optimal stable matching and the man-optimal stable matching for $p$;
by $PO$ the matching mechanism that associates with any $p\in \mathcal{P}$ the set of matchings that are Pareto optimal for $p$;
by $WPO$ the matching mechanism that associates with any $p\in \mathcal{P}$ the set of matchings that are weakly Pareto optimal for $p$;
by $MO$ the matching mechanism that associates with any $p\in \mathcal{P}$ the set of the matchings that are minimally optimal for $p$.\footnote{$TO$, $ST$, $GS$, $PO$, $WPO$ and $MO$ stand for \lq\lq Totality\rq\rq, \lq\lq Stable\rq\rq, \lq\lq Gale and Shapley\rq\rq, \lq\lq Pareto Optimal\rq\rq,  \lq\lq Weakly Pareto Optimal\rq\rq,  \lq\lq Minimally Optimal\rq\rq,\ respectively.}

Let us introduce some basic definitions.

\begin{definition}
{\rm Let $F$ be a matching mechanism on $\mathcal{P}$. We say that $F$ is:
\begin{itemize}
\item {\it decisive} if, for every $p\in \mathcal{P}$, $F(p)\neq\varnothing$;
\item {\it resolute} if, for every $p\in \mathcal{P}$, $|F(p)|=1$.
\end{itemize}}
\end{definition}

Of course, if $F$ is resolute, then $F$ is decisive. Note that $GS$, $ST$, $PO$, $WPO$, $MO$ and $TO$ are decisive but they are not resolute; on the other hand, $GS_w$ and $GS_m$ are resolute (Gale and Shapley, 1962).

\begin{definition}
{\rm Let $F$ and $F'$ be matching mechanisms. We say that $F'$ is a {\it refinement} of $F$ if, for every $p\in \mathcal{P}$, $F'(p)\subseteq F(p)$. If $F'$ is a refinement of $F$ we write $F'\subseteq F$.}
\end{definition}

Note that if $F$, $F'$ and $F''$ are matching mechanisms, then $F''\subseteq F'$ and $F'\subseteq F$ imply $F''\subseteq F$. Moreover, we have that $GS\subseteq ST\subseteq PO\subseteq WPO\subseteq MO\subseteq TO$, $GS_w\subseteq GS$ and $GS_m\subseteq GS$.

Given a matching mechanism $F$, it is natural to interpret any refinement of $F$ as a way to reduce the ambiguity of $F$. Resolute refinements of $F$ are particularly important as they allow to completely eliminate the ambiguity of $F$. Of course, any decisive matching mechanism admits in general many resolute refinements, so it becomes important to understand whether some of them fulfill desirable properties.

The next definition presents some properties a matching mechanism may have.

\begin{definition}\label{mm-prop}
{\rm Let $F$ be a matching mechanism. We say that $F$ is
{\it stable} [{\it Pareto optimal; weakly Pareto optimal;  minimally optimal}\,] if, for every $p\in \mathcal{P}$ and $\mu\in F(p)$, $\mu$ is stable [Pareto optimal; weakly Pareto optimal;  minimally optimal] for $p$.}
\end{definition}

Note that a matching mechanism $F$ is stable [Pareto optimal; weakly Pareto optimal; minimally optimal] if and only if $F$ is a refinement of $ST$ [$PO$; $WPO$; $MO$].

\subsection{Symmetric matching mechanisms} \label{SymMM}

In this section we define the key concept of our paper, that is, the notion of symmetric matching mechanism. Some preliminaries are needed before providing such definition. First, we introduce some notation in order to model changes in individuals' identities. To this aim, we consider the  subsets of $\mathrm{Sym}(I)$ given by
\[
G^*\coloneq\{\varphi\in \mathrm{Sym}(I): \{\varphi(W)\,,\varphi(M)\}=\{W,M\}\},
\]
and
\[
G\coloneq\{\varphi\in \mathrm{Sym}(I): \varphi(W)=W,\,\varphi(M)=M\}.
\]
These two sets will play a main role throughout the paper. $G^*$ is formed by the permutations over $I$ that keep the partition $\{W,M\}$ of $I$ fixed; these permutations model changes in individuals' identities either within the sets $W$ and $M$ (i.e. women are renamed as women and men are renamed as men) or across them (i.e. women are renamed as men and men are renamed as women). $G$ is the subset of $G^*$ formed by the permutations over $I$ that transform the set $W$ into itself and the set $M$ into itself; those permutations only model changes in individuals' identities within the sets $W$ and $M$. Note that, as a consequence, the permutations in $G^*\setminus G$ rename each woman with one of men's names and each man with one of women's names. We can also distinguish two further subsets of $G$, namely
\[
G_W\coloneq\{\varphi\in \mathrm{Sym}(I):  \varphi(y)=y \mbox{ for all }y\in M\},
\]
\[
G_M\coloneq\{\varphi\in  \mathrm{Sym}(I):  \varphi(x)=x \mbox{ for all }x\in W\}.
\]
$G_W$ consists of those permutations of individuals' names that leave men's names unchanged; $G_M$ consists of those permutations of individuals' names that leave women's names unchanged.

For instance, if $W=\{1,2\}$ and $M=\{3,4\}$, we have that
\begin{equation}\label{Gstar2}
G^*=\{id_I,(12),(34),(12)(34),(13)(24),(14)(23),(1324),(1423)\},
\end{equation}
\[
G=\{id_I,(12),(34),(12)(34)\},\quad  G_W=\{id_I,(12)\},\quad G_M=\{id_I,(34)\}.
\]
Given $p\in \mathcal{P}$ and $\varphi\in G^*$, taking into account  \eqref{prodottorel}, we denote by $p^\varphi$ the preference profile
defined, for every $z\in I$, by
\begin{equation}\label{action-w}
p^\varphi(z)=\varphi p(\varphi^{-1}(z)).
\end{equation}
By Lemma \ref{phiRR}, it is easily checked that  $p^\varphi$ actually belongs to $\mathcal{P}$.
Moreover, \eqref{action-w} is equivalent to state that, for every $z\in I$,
\begin{equation}\label{action-w2}
p^\varphi(\varphi(z))=\varphi p(z),
\end{equation}
which is more intuitive than \eqref{action-w}.
Indeed, consider $\varphi \in G^*$ and $p\in\mathcal{P}$ represented by the table
\[
\begin{array}{|ccc||ccc|}
\hline
1&\ldots&n&n+1&\ldots&2n\\
\hline
\hline
p(1)&\ldots&p(n)&p(n+1)&\ldots&p(2n)\\
\hline
\end{array}.
\]
Then, by \eqref{action-w2}, we deduce that
$p^\varphi$ is represented by the table
\[
\begin{array}{|ccc||ccc|}
\hline
\varphi(1)&\ldots&\varphi(n)&\varphi(n+1)&\ldots&\varphi(2n)\\
\hline
\hline
\varphi p(1)&\ldots&\varphi p(n)&\varphi p(n+1)&\ldots&\varphi p(2n)\\
\hline
\end{array}.
\]
For a concrete example, consider the preference profile $p$ represented by the table in \eqref{tablep}. If $\varphi=(123)(46)\in G$, then
\[
p^{\varphi}=
\begin{array}{|ccc||ccc|}
\hline
2&3&1&6&5&4\\
\hline
\hline
6&6&4&3&1&1\\
5&4&5&2&2&3\\
4&5&6&1&3&2\\
\hline
\end{array}
=
\begin{array}{|ccc||ccc|}
\hline
1&2&3&4&5&6\\
\hline
\hline
4&6&6&1&1&3\\
5&5&4&3&2&2\\
6&4&5&2&3&1\\
\hline
\end{array};
\]
if  $\varphi=(1426)(35)\in G^*$, then
\[
p^{\varphi}=
\begin{array}{|ccc||ccc|}
\hline
4&6&5&2&3&1\\
\hline
\hline
2&2&1&6&5&5\\
3&1&3&4&4&6\\
1&3&2&5&6&4\\
\hline
\end{array}
=
\begin{array}{|ccc||ccc|}
\hline
1&2&3&4&5&6\\
\hline
\hline
5&6&5&2&1&2\\
6&4&4&3&3&1\\
4&5&6&1&2&3\\
\hline
\end{array}.
\]

We introduce the notion of conjugate in order to define the main concept of the paper. Let $H$ be a group and $h\in H$. The conjugate of $k\in H$ by $h$ is defined by $hkh^{-1}\in H$ and denoted by $k^h$; the conjugate of $K\subseteq H$ by $h$ is defined by $hKh^{-1}\coloneq\{hkh^{-1}\in H:k\in K\}$ and denoted by $K^h$.
 Note that, for every $k,h_1,h_2\in H$ and $K\subseteq H$, we have that 
 \begin{equation}\label{conprop}
 (k^{h_1})^{h_2}=k^{h_2h_1}\quad \hbox{and}\quad  (K^{h_1})^{h_2}=K^{h_2h_1}.
 \end{equation}
Moreover, for every $h\in H$ and $K\subseteq H$, the function from $K$ to $K^h$ that associates $k^h$ with any $k\in K$ is a bijection.

Consider, for instance, $(123)(45), (253)\in \mathrm{Sym}(\ldbrack 5\rdbrack)$. Then, an easy computation shows that $[(123)(45)]^{(253)}$, that is the conjugate of $(123)(45)$ by $(253)$, is given by $(152)(43)$. As one can notice, the permutation $(152)(43)$ is obtained by $(123)(45)$ by replacing each element of each cycle by its image through $(253)$. That is a general fact. Indeed, if a permutation $\sigma$ is a product of the permutations $\sigma_1,\dots, \sigma_k$ and $\varphi$ is a further permutation, then
\begin{equation}\label{conj-split}
(\sigma_1\cdots\sigma_k)^{\varphi}=\sigma_1^{\varphi}\cdots\sigma_k^{\varphi}.\footnote{In the language of group theory, we are just expressing the well-known fact  that conjugation is a homomorphism. }
\end{equation}
Moreover, if $\sigma=(x_1\cdots x_m)$ is a $m$-cycle, then
\begin{equation}\label{conj-cyc}
\sigma^{\varphi}=(\varphi(x_1)\cdots \varphi(x_m)).
\end{equation}
In particular, the cycle structure of a permutation, that is the number of cycles in which it splits and their lengths, is invariant by conjugation.  For further details on the computation of conjugates in symmetric groups, see Example 4.29 and Chapter $4$ in Milne (2021).
It is worth noticing that, for every $\varphi\in G^*$,  $\mathcal{M}^{\varphi}=\mathcal{M}$ (see Proposition \ref{car-match} in Appendix \ref{GT2}).

We are now ready to provide the definition of symmetric matching mechanism.
\begin{definition}\label{usymm}
{\rm Let $F$ be a matching mechanism and $U\subseteq G^*$. We say that $F$ is $U$-{\it symmetric} if, for every $p\in \mathcal{P}$ and $\varphi\in U$, $F(p^{\varphi})= F(p)^{\varphi}$.
If $F$ is $G^*$-symmetric, we simply say that $F$ is symmetric.}
\end{definition}
In order to describe the intuition behind the definition, assume that $W=\{1,2,3\}$ and $M=\{4,5,6\}$ and consider $U\subseteq G^*$ and a $U$-symmetric matching mechanism $F$. Assume that $\varphi=(1426)(35)\in U$ and that 
\[
F(p)=\{(14)(25)(3 6),(1 5)(2 6)(3 4)\}.
\]
Then, by Definition \ref{usymm}, and using \eqref{conj-split} and \eqref{conj-cyc}, we get
\[
F(p^\varphi)=F(p)^{\varphi}=\big\{[(1 4)(2 5)(3 6)]^{\varphi}, [(1 5)(2 6)(3 4)]^{\varphi}\big\}
=\{(4 2)(6 3)(5 1), (4 3)(6 1)(5 2)\}.
\]
The above example clearly shows that, if $F$ is $U$-symmetric, a change in individuals' names in the input according to a permutation in $U$ entails the same change of names in the matchings of the output. This means that  the notion of $U$-symmetry models a notion of  fairness, which may assume various shapes depending on the choice of $U$.
For instance, if $F$ is $G_W$-symmetric, then women are equally treated among themselves by $F$; if $F$ is $G_M$-symmetric, then men are equally treated among themselves by $F$; if $F$ is $G$-symmetric, then women are equally treated among themselves and men are equally treated among themselves by $F$; if $F$ is symmetric, then women are equally treated among themselves, men are equally treated among themselves and the two groups of women and men are equally treated by $F$.

Definition \ref{usymm} provides a generalization of several properties analyzed in the literature on matching mechanisms. Indeed, the standard concept of anonymity, called peer indifference by Masarani and Gokturk (1989), coincides with $G$-symmetry; the property of gender indifference introduced by Masarani and Gokturk (1989) coincides with  $\{\varphi\}$-symmetry, where $\varphi$ is a suitable element of $\mathcal{M}$ (see also Endriss, 2020); the property of gender fairness by \"{O}zkal-Sanver (2004) coincides with $(G^*\setminus G)$-symmetry. Of course, a matching mechanism $F$ satisfies both peer indifference and gender indifference if and only if $F$ is  $(G\cup\{\varphi\})$-symmetric.

\section{Research problem and main results}\label{reserach-problem}

This section is divided into three parts. We first illustrate and motivate our research problem; then, we present an existence condition which is fundamental for proving the possibility and impossibility theorems that represent the main results of our paper; in the last part of the section we collect those results.

\subsection{The research problem} \label{ResProb}

Proving or disproving the existence of resolute and symmetric matching mechanisms, possibly satisfying further properties, is certainly an interesting and not trivial problem. Since a resolute and symmetric matching mechanism can be seen as a resolute and symmetric refinement of $TO$ (that is, the matching mechanism that associates with any preference profile the  whole set $\mathcal{M}$ of matchings), and since many properties of matching mechanisms can be described using the concept of refinement (e.g. stability, Pareto optimality, weak Pareto optimality, minimal optimality), the aforementioned existence problem can be fruitfully specialized in the following existence problem.

\vspace{2mm}

\noindent {\bf Main Problem.} Let $U \subseteq G^*$ and $F$ be a matching mechanism.
Find conditions on $U$ and $F$ that guarantee the existence of a resolute refinement of $F$ that is $U$-symmetric.
\vspace{2mm}

Unfortunately, a strong tension between resoluteness, symmetry and optimality requirements emerges. Indeed, $GS_w$ and $GS_m$ are resolute and $G$-symmetric but fail to be symmetric; on the other hand, $GS$, $ST$, $PO$, $WPO$, $MO$ and $TO$ are symmetric but not resolute (see Proposition \ref{STsym} in  Appendix \ref{GT2} for the proof of the symmetry properties of the aforementioned matching mechanisms).

Several stable matching mechanisms have been introduced in the literature with the aim to guarantee more equity between the two sets $W$ and $M$ than the one obtained by $GS_w$ and $GS_m$. The Sex-Equal mechanism, introduced by Gusfield and Irving (1989) and analyzed in Romero-Medina (2001), and the Egalitarian Stable mechanism, introduced in Gusfield and Irving (1989), are two of them. Roughly speaking, both of them are based on the idea to quantify first, for every stable matching, the envy perceived by every individual and then consider the stable matchings minimizing a suitable function involving the level of envy of all the individuals. Given a preference profile $p$, a matching $\mu$ and $z\in I$, we have that the number $\mathrm{Rank}_{p(z)}(\mu(z))-1$ counts the number of individuals envied by $z$ under the matching $\mu$. Indeed, consider for instance, $x \in W$. Then $x$ envies the partners of all the men $y\in M$ such that $y \succ_{p(x)} \mu(x)$ and their number is exactly $\mathrm{Rank}_{p(x)}(\mu(x))-1$. For every $p\in\mathcal{P}$ and $\mu\in\mathcal{M}$, we can consider then the quantities
\begin{equation*}
\delta(p,\mu) \coloneq \left|\sum_{x \in W}\mathrm{Rank}_{p(x)}(\mu(x))-\sum_{y \in M}\mathrm{Rank}_{p(y)}(\mu(y))\right|,
\end{equation*}
and
\begin{equation*}
e(p,\mu) \coloneq \sum_{x \in W}\mathrm{Rank}_{p(x)}(\mu(x))+\sum_{y \in M}\mathrm{Rank}_{p(y)}(\mu(y)).
\end{equation*}
Since
\[
\sum_{x \in W}\mathrm{Rank}_{p(x)}(\mu(x)) \quad \mbox{ and }\quad\sum_{y \in M}\mathrm{Rank}_{p(y)}(\mu(y))
\]
basically represent the aggregate level of envy of the men and the aggregate level of envy of the women, we have that
$\delta(\mu,p)$ measure the distance between those two values while $e(\mu, p)$ measures the overall level of envy in the society.
The Sex-Equal matching mechanism, denoted by $SE$, is the matching mechanism defined, for every $p\in\mathcal{P}$, by
\[
SE(p) \coloneq \underset{\mu \in ST(p)}{\mathrm{arg\,min}}\; \delta(p,\mu);
\]
the Egalitarian Stable matching mechanism, denoted by $ES$, is the matching mechanism defined, for every $p\in\mathcal{P}$, by
\[
ES(p) \coloneq \underset{\mu \in ST(p)}{\mathrm{arg\,min}}\; e(p,\mu).
\]
Of course, $SE$ and $ES$ are decisive and stable. Moreover, they are also symmetric, as proved in Proposition \ref{SEsym} in Appendix \ref{GT2}.
However, by Example 1 in Romero-Medina (2001, page 201), they are not resolute.

\subsection{A preliminary existence condition} \label{PreEx}

Here, we present a general existence result for $U$-symmetric resolute refinements of $U$-symmetric matching mechanisms (Theorem \ref{F-ref-consistent}). Some preliminary notions and results are needed for stating such theorem. They are illustrated in what follows.

Let $H$ be a finite group and $K\subseteq H$. We say that $K$ is a subgroup of $H$, and we write $K\leq H$, if the restriction of the operation  on $H$ to the set $K\times K$ produces a group structure for $K$. It is well known that, since $H$ is finite, $K\le H$ holds  if and only if, $K\neq\varnothing$ and for every $h,k\in K$, the product $hk$ belongs to $K$. Clearly, if $K'\le K$ and $K\le H$ then $K'\le H$.
The subgroup of $H$ generated by $K$, denoted by $\langle K\rangle$, is the intersection of all the subgroups of $H$ containing $K$. For instance, $\langle \varnothing\rangle=\{1\}$. If $h\in H$, we write  $\langle h\rangle$ instead of $\langle\{ h\}\rangle$.
Of course, if $K\le H$, then $\langle K\rangle=K$.
It is well known that, since $H$ is finite, when $K\ne \varnothing$, $\langle K\rangle$ is the set made up by all the elements of $K$ and the products of a finite number of elements in $K$. For further details on the concepts and statements above, see Chapter 1 in Milne (2021). For our purposes it is important to observe that $G\le G^*\le \mathrm{Sym}(X)$, $G_W\le G$ and $G_M\le G$. Moreover, we have $|G|=(n!)^2$ and $|G^*|=2(n!)^2$. Those facts are easily checked.

As a first result, we show that dealing with symmetry for subsets of $G^*$ is equivalent to dealing with symmetry for subgroups of $G^*$.
The proof of Proposition 2 is in Appendix \ref{GT2}.

\begin{proposition}\label{iff}
Let $F$ be a matching mechanism and $U\subseteq G^*$.
Then $F$ is $U$-symmetric if and only if $F$ is $\langle U\rangle$-symmetric.
\end{proposition}

Proposition \ref{iff} is remarkable and emblematic of the role that group theory may play in the context of matching theory. In fact, it points out that whenever one focuses on the symmetry with respect to a set of permutations, it inevitably ends up with considering the symmetry with respect to a subgroup of permutations. As a consequence, while considering various levels of symmetry, we are allowed to restrict our attentions to subgroups of $G^*$ only. As a further consequence, since $\langle G^*\setminus G\rangle=\langle G\cup\{\varphi\}\rangle=G^*$ (see Propositions \ref{generato1} and \ref{generato2} in Appendix \ref{GT2}), we deduce that the property of gender fairness coincides with $G^*$-symmetry and that the properties of peer indifference and gender indifference together are equivalent to $G^*$-symmetry. We emphasize that, to the best of our knowledge, the fact that requiring gender fairness is the same as requiring both peer indifference and gender indifference does not seem to be present in the literature.

Let us introduce another concept that will turn out very important for our research.
\begin{definition}
{\rm Let $U\le G^*$ and $p\in\mathcal{P}$. The $U$-{\it stabilizer} of $p$ is the set
\[
\mathrm{Stab}_U(p)\coloneq\left \{\varphi\in U: p^\varphi=p \right \}.
\]}
\end{definition}
Remarkably, $\mathrm{Stab}_U(p)$ is a subgroup of $U$ and, for every $\varphi\in U$, we have $\mathrm{Stab}_U(p^{\varphi})=[\mathrm{Stab}_U(p)]^{\varphi}$.
Moreover, it is immediately checked that $\mathrm{Stab}_U(p)=\mathrm{Stab}_{G^*}(p)\cap U$. That explains, in particular, the central role of the $G^*$-stabilizers. Indeed, the knowledge of the $G^*$-stabilizer of a preference profile allows to know its $U$-stabilizer for all $U\le G^*$.

As an easy example, let $W=\{1,2\}$, $M=\{3,4\}$ and $p$ be the preference profile represented by the table
 \begin{equation}\label{pspeciale}
p=
\begin{array}{|cc||cc|}
\hline
1&2&3&4\\
\hline
\hline
3&4&2&1\\
4&3&1&2\\
\hline
\end{array}.
\end{equation} 
Using \eqref{Gstar2}, it turns out that $\mathrm{Stab}_{G^*}(p)=\{id_I,(1324),(12)(34),(1423)\}$. As a consequence, if $U\coloneq\{id_I, (12)(34)\}$, we find $\mathrm{Stab}_{U}(p)=\{id_I,(12)(34)\}$; if instead $U\coloneq\{id_I, (12)\}$, we find $\mathrm{Stab}_{U}(p)=\{id_I\}.$

In the next definition we introduce, for every $U\le G^*$, a matching mechanism denoted by $C^U$. Such a mechanism plays a crucial role to address the aforementioned Main Problem, as shown by all the subsequent results of the section, especially by Theorem \ref{F-ref-consistent}.

\begin{definition}
{\rm Let $U\le G^*$. We denote by $C^U$ the matching mechanism defined, for every $p\in\mathcal{P}$, by 
\[
C^{U}(p)\coloneq \left\{\mu\in \mathcal{M}: \forall \varphi\in \mathrm{Stab}_U(p),\; \mu^\varphi =\mu\right\}.
\]}
\end{definition}
The letter $C$ intentionally resembles the word centralizer largely used by group theorists.
Let $H$ be a group and $K\le H$. The centralizer of $K$ in $H$ is defined by
\[
C_{H}(K)\coloneq\{h\in H: \forall k \in K,\;kh=hk\}.
\]
It is well known that $C_{H}(K)$ is a subgroup of $H$ and that, for every $h\in H$, we have that $C_{H}(K^h)=[C_{H}(K)]^h$.

It is then clear that, given $U\le G^*$ and $p\in\mathcal{P}$, 
\begin{equation}\label{nuova-veste}
C^{U}(p)=\left\{\mu\in \mathcal{M}: \forall \varphi\in \mathrm{Stab}_U(p),\; \varphi\mu=\mu \varphi\right\}=C_{G^*}( \mathrm{Stab}_U(p))\cap \mathcal{M}.
\end{equation}
Note that, since $\mathcal{M}$ is not a subgroup of $G^*$, in general, $C^{U}(p)$ is not a subgroup of $G^*$. 
We now illustrate the properties of $C^U$. The purpose is, on the one hand, to better understand the new concept and, on the other hand, to communicate a more vivid intuition of it.
\begin{proposition}\label{facile1} Let $U\le G^*$. Then the following facts hold true:
\begin{itemize}
\item[$(i)$] $C^U$ is $U$-symmetric; 
\item[$(ii)$] if $F$ is a  resolute and $U$-symmetric matching mechanism, then $F$ is a refinement of $C^U$.
\end{itemize}
\end{proposition}

\begin{proof} 
$(i)$ We have to show that, for every $p\in \mathcal{P}$ and $\varphi\in U$, $C^U(p^{\varphi})= C^U(p)^{\varphi}$. 
Consider then $p\in \mathcal{P}$ and $\varphi\in U$. Recalling that conjugation defines a bijection and using \eqref{nuova-veste} and the properties of stabilizers and centralizers, we have that
\[
C^{U}(p^{\varphi})=C_{G^*}( \mathrm{Stab}_U(p^{\varphi}))\cap \mathcal{M}=C_{G^*}( [\mathrm{Stab}_U(p)]^{\varphi})\cap \mathcal{M}^{\varphi}
\]
\[
=(C_{G^*} [\mathrm{Stab}_U(p)])^{\varphi}\cap \mathcal{M}^{\varphi}
=(C_{G^*} [\mathrm{Stab}_U(p)]\cap \mathcal{M})^{\varphi}=C^U(p)^{\varphi}.
\]

$(ii)$ Let $F$ be a resolute and $U$-symmetric matching mechanism. Consider $p\in \mathcal{P}$ and let $\mu$ be the unique matching in $F(p)$. By \eqref{nuova-veste}, we prove that $\mu\in C^U(p)$ by showing that, for every $\varphi\in \mathrm{Stab}_U(p)$, the equality $ \mu^\varphi=\mu$ holds. Let $\varphi\in \mathrm{Stab}_U(p)$. Then, $p^\varphi=p$ and, since $F$ is $U$-symmetric, we have
$$
\{\mu\}=F(p)=F(p^\varphi)=F(p)^\varphi=\{\mu\}^\varphi=\{\mu^\varphi\},
$$
that implies the desired equality.
\end{proof}
By the previous proposition, it turns out that $C^U$ is the natural container for all the resolute and $U$-symmetric matching mechanisms. Indeed, the output of a resolute and $U$-symmetric matching mechanism on a preference profile $p$ necessarily belongs to $C^U(p)$.

We are now ready to state the main result of the section. Its proof is in Appendix \ref{Appendix-orbits}.

\begin{theorem}\label{F-ref-consistent}
Let $U\le G^*$ and $F$ be a $U$-symmetric matching mechanism. Then, $F$ admits a $U$-symmetric and resolute refinement if and only if, for every $p\in \mathcal{P}$, $F(p)\cap C^{U}(p)\neq\varnothing$.
\end{theorem}

From Theorem \ref{F-ref-consistent}, we immediately deduce that a necessary condition for the existence of a $U$-symmetric and resolute refinement of a $U$-symmetric matching mechanism is the decisiveness of $C^{U}$. Unfortunately, $C^{U}$ may fail to be decisive for some values of the size of $W$ and $M$ and for some subgroups $U$ of $G^*$. Assume, for instance, $W=\{1,2\}$, $M=\{3,4\}$, and $U=G^*$. If $p$ is the preference profile in \eqref{pspeciale}, we know that $\mathrm{Stab}_{G^*}(p)=\{id_I,(1324),(12)(34),(1423)\}$. Using \eqref{Gstar2}, a direct computation shows that $C_{G^*}( \mathrm{Stab}_{G^*}(p))=\mathrm{Stab}_{G^*}(p)$. As a consequence, we have $C^{G^*}(p)=C_{G^*}( \mathrm{Stab}_{G^*}(p))\cap \mathcal{M}=\mathrm{Stab}_{G^*}(p)\cap \mathcal{M}=\varnothing$.

\subsection{Possibility and impossibility results} \label{PossImpossThs}

We are now ready to state the possibility and impossibility results proved in the paper. The proofs of the results collected in this section are in Appendix \ref{main-proofs}. Theorems \ref{n2} and \ref{main-mathc} are about $G^*$-symmetry and are the main results of the paper. We stress that the proof of Theorem \ref{main-mathc} relies on several results and facts from group theory, some of them definitely non-trivial. Recall that $n$ denotes the size of $W$ and $M$.

\begin{theorem}\label{n2}
Let $n$ be even. Then there exists no resolute and symmetric matching mechanism.
\end{theorem}

\begin{theorem}\label{main-mathc}
Let $n$ be odd. Then there exists a resolute and symmetric matching mechanism.
\end{theorem}

Given the possibility result stated in Theorem \ref{main-mathc}, one may wonder whether it is possible, for the case $n$ odd,  to be more ambitious and ask for some further property, like minimal optimality, and still get the existence.  The next result shows that the answer is negative.

\begin{theorem}\label{no-sym-mo}
Let $n$ be odd. Then there exists no resolute, symmetric and minimally optimal matching mechanism.
\end{theorem}

Taking into account Theorems \ref{n2} and \ref{no-sym-mo}, we then get the following general impossibility result.

\begin{theorem}\label{impossibilita1}
There exists no resolute, symmetric and minimally optimal matching mechanism.
\end{theorem}

Of course, as a consequence of the previous theorem and Proposition \ref{implications}, we also deduce that there exists no resolute, symmetric and stable [Pareto optimal; weakly Pareto optimal] matching mechanism. 

Now we move our focus on the weaker notion of $\{\varphi\}$-symmetry, where  $\varphi\in\mathcal{M}$. 
As established by Proposition \ref{iff},  $\{\varphi\}$-symmetry is equivalent to $\langle \varphi\rangle$-symmetry and we adopt the latter writing from now on.  Clearly $\varphi$ exchanges individuals' identities between the two groups $W$ and $M$.
Since $\langle \varphi\rangle$ is a subgroup of $G^*$ of order $2$, the only smaller subgroup of $\langle \varphi\rangle$ is given by the trivial group $\{id_I\}$. Thus there exists no proper subgroup of $\langle \varphi\rangle$ containing permutations exchanging individuals' identities between the two groups $W$ and $M$. As a consequence, $\langle\varphi\rangle$-symmetry represents a minimal level of symmetry among those that allow for an exchange of individuals' identities between the two groups $W$ and $M$.

The first result we propose is a possibility result.

\begin{theorem}\label{existence3}
Let $\varphi\in\mathcal{M}$. Then there exists a resolute, $\langle \varphi\rangle$-symmetric and weak Pareto optimal matching mechanism.
\end{theorem}

Given Theorem \ref{existence3}, one may wonder whether it is possible to replace the assumption of weak Pareto optimality with the stronger property of stability and still get existence. The next two results show that this can be actually done if and only if $n=2$.
Note that Theorem \ref{endriss} corresponds to Theorem 6 in Endriss (2020).\footnote{Observe also that Theorem \ref{endriss} implies Theorem 3 in Root and Bade (2023).} The proof in Endriss is based both on first-order logic and on Satisfiability (SAT) solving techniques. In fact, it is first shown that some properties of matching mechanisms, including stability and $\langle \varphi\rangle$-symmetry, can be encoded in a specific formal language. Then, it is proved that impossibility results involving these properties can be reduced to impossibility results for the case with three women and three men and that the analysis of this latter case can be fully automated using SAT solving technology. Here, we propose an alternative proof of that impossibility result.

\begin{theorem}\label{existence2}
Let $n=2$ and $\varphi\in\mathcal{M}$. Then there exists a resolute,  $\langle \varphi\rangle$-symmetric and stable matching mechanism.
\end{theorem}

\begin{theorem}[Endriss, 2020]\label{endriss}
Let $n\ge 3$ and $\varphi\in\mathcal{M}$. Then there exists no resolute, $\langle \varphi\rangle$-symmetric and stable  matching mechanism.
\end{theorem}

Of course, in light of Theorem \ref{existence3}, one may also wonder whether there are resolute, $\langle \varphi\rangle$-symmetric
and Pareto optimal matching mechanism. While, by Theorem \ref{existence2}, we get existence if $n=2$, when $n\ge 3$ the problem is still open.

\section{A generalization: the model with outside option}\label{outside}

The model studied in this paper considers those situations where individuals are not allowed to stay alone, but are always forced to be matched to another individual. The model where the outside option of being unmatched is given to individuals is important and well studied in the literature. In this section, after recalling how preference profiles, matching and matching mechanisms are defined in the model with outside option, we introduce a concept of symmetry analogous to the one in Definition \ref{usymm}. Then, we propose an impossibility result.

A generalized preference profile is a function $\overline{p}$ from $I$ to the set
\[
\left(\bigcup_{x\in W}\mathbf{L}(M\cup\{x\})\right)\cup \left(\bigcup_{y\in M}\mathbf{L}(W\cup\{y\})\right)
\]
such that, for every $x\in W$, $\overline{p}(x)\in \mathbf{L}(M\cup\{x\})$ and, for every $y\in M$, $\overline{p}(y)\in \mathbf{L}(W\cup\{y\})$.
The set of generalized preference profiles is denoted by $\overline{\mathcal{P}}$.
We represent a generalized preference profile $\overline{p}\in\overline{\mathcal{P}}$ by the table
\[
\begin{array}{|ccc||ccc|}
\hline
1&\ldots&n&n+1&\ldots&2n\\
\hline
\hline
\overline{p}(1)&\ldots&\overline{p}(n)&\overline{p}(n+1)&\ldots&\overline{p}(2n)\\
\hline
\end{array},
\]
where in the first row there are the names of individuals and in the second row the corresponding preferences.
For instance, if $W=\{1,2,3\}$ and $M=\{4,5,6\}$, the table
\begin{equation}\label{tablepgen}
\begin{array}{|ccc||ccc|}
\hline
1&2&3&4&5&6\\
\hline
\hline
4&4&6&2&3&6\\
1&6&5&1&5&2\\
6&2&3&3&2&1\\
5&5&4&4&1&3\\
\hline
\end{array}
\end{equation}
represents the preference profile $\overline{p}$ such that
\[
\overline{p}(1)=[4,1,6,5]\in \mathbf{L}(M\cup\{1\}), \; \overline{p}(2)=[4,6,2,5]\in \mathbf{L}(M\cup\{2\}), \; \overline{p}(3)=[6,5,3,4]\in \mathbf{L}(M\cup\{3\}),
\]
\[
\overline{p}(4)=[2,1,3,4]\in \mathbf{L}(W\cup\{4\}), \; \overline{p}(5)=[3,5,2,1]\in \mathbf{L}(W\cup\{5\}), \; \overline{p}(6)=[6,2,1,3]\in \mathbf{L}(W\cup\{6\}).
\]
Given $\overline{p}\in \overline{\mathcal{P}}$ and $\varphi\in G^*$, taking into account the definition \eqref{prodottorel}, we denote by $\overline{p}^\varphi$ the generalized preference profile defined, for every $z\in I$, by $\overline{p}^\varphi(z)=\varphi \overline{p}(\varphi^{-1}(z))$. Using Lemma \ref{phiRR}, it is easily checked that $\overline{p}^\varphi$ actually belongs to $\overline{\mathcal{P}}$. Similarly to what said for preference profiles in Section \ref{SymMM}, if $\varphi \in G^*$ and $\overline{p}\in\overline{\mathcal{P}}$ is represented by the table
\[
\begin{array}{|ccc||ccc|}
\hline
1&\ldots&n&n+1&\ldots&2n\\
\hline
\hline
\overline{p}(1)&\ldots&\overline{p}(n)&\overline{p}(n+1)&\ldots&\overline{p}(2n)\\
\hline
\end{array},
\]
then $\overline{p}^\varphi$ is represented by the table
\[
\begin{array}{|ccc||ccc|}
\hline
\varphi(1)&\ldots&\varphi(n)&\varphi(n+1)&\ldots&\varphi(2n)\\
\hline
\hline
\varphi \overline{p}(1)&\ldots&\varphi \overline{p}(n)&\varphi \overline{p}(n+1)&\ldots&\varphi \overline{p}(2n)\\
\hline
\end{array}.
\]
Thus, $\overline{p}^\varphi$ is obtained by $\overline{p}$ as if all individual identities were changed according to $\varphi$.
For a concrete example, consider the generalized preference profile $\overline{p}$ represented by the table in \eqref{tablepgen}. If $\varphi=(123)(46)\in G$, then
\[
\overline{p}^{\varphi}=
\begin{array}{|ccc||ccc|}
\hline
2&3&1&6&5&4\\
\hline
\hline
6&6&4&3&1&4\\
2&4&5&2&5&3\\
4&3&1&1&3&2\\
5&5&6&6&2&1\\
\hline
\end{array}
=
\begin{array}{|ccc||ccc|}
\hline
1&2&3&4&5&6\\
\hline
\hline
4&6&6&4&1&3\\
5&2&4&3&5&2\\
1&4&3&2&3&1\\
6&5&5&1&2&6\\
\hline
\end{array};
\]
if  $\varphi=(1426)(35)\in G^*$, then
\[
\overline{p}^{\varphi}=
\begin{array}{|ccc||ccc|}
\hline
4&6&5&2&3&1\\
\hline
\hline
2&2&1&6&5&1\\
4&1&3&4&3&6\\
1&6&5&5&6&4\\
3&3&2&2&4&5\\
\hline
\end{array}
=
\begin{array}{|ccc||ccc|}
\hline
1&2&3&4&5&6\\
\hline
\hline
1&6&5&2&1&2\\
6&4&4&4&3&1\\
4&5&6&1&5&6\\
5&2&4&3&2&3\\
\hline
\end{array}.
\]

A generalized matching is a permutation $\overline{\mu}\in \mathrm{Sym}(I)$ such that, for every $x\in W$, if $\overline{\mu}(x)\neq x$ then  $\overline{\mu}(x)\in M$; for every $y\in M$, if $\overline{\mu}(y)\neq y$ then  $\overline{\mu}(y)\in W$; for every $z\in I$, $\overline{\mu}(\overline{\mu}(z))=z$. The set of generalized matchings is denoted by $\overline{\mathcal{M}}$. For instance, if $W=\{1,2\}$ and $M=\{3,4\}$, then 
\[
\overline{\mathcal{M}}=\Big\{
id_I, (13),(14),(23),(24),(13)(24),(14)(23)\Big\}.
\]
Recalling \eqref{conj-split} and \eqref{conj-cyc}, it is clear that, for every $\overline{\mu}\in\overline{\mathcal{M}}$ and $\varphi\in G^*$, we have $\overline{\mu}^\varphi\in\overline{\mathcal{M}}$. 
Note also that, $\overline{\mu}^\varphi$ is obtained by $\overline{\mu}$ as if all individual identities were changed according to $\varphi$. As an example, assume $W=\{1,2,3\}$ and $M=\{4,5,6\}$ and consider $\overline{\mu}=(15)(26)\in\overline{\mathcal{M}}$ and $\varphi=(1426)(35)\in G^*$. Then $\overline{\mu}^\varphi=(43)(61)$.

Given $\overline{p}\in \overline{\mathcal{P}}$  and $\overline{\mu}\in\overline{\mathcal{M}}$, we say that $\overline{\mu}$ is
stable for $\overline{p}$ if there is no $z\in W\cup M$ such that $z \succ_{\overline{p}(z)} \overline{\mu}(z)$ and there is no $(x,y)\in W\times M$ such that $y\succ_{\overline{p}(x)} \overline{\mu}(x)$ and $x\succ_{\overline{p}(y)} \overline{\mu}(y)$; Pareto optimal for $\overline{p}$ if there is no $\overline{\mu}'\in\overline{\mathcal{M}}$ such that, for every $z \in I$, $\overline{\mu}'(z) \succeq_{\overline{p}(z)} \overline{\mu}(z)$ and there exists $z^* \in I$ such that $\overline{\mu}'(z^*) \succ_{\overline{p}(z^*)} \overline{\mu}(z^*)$. It can be easily shown that if $\overline{\mu}$ is
stable for $\overline{p}$, then it is Pareto optimal for $\overline{p}$. 

A (two-sided) generalized matching mechanism is a correspondence from $\overline{\mathcal{P}}$ to $\overline{\mathcal{M}}$.
Given a generalized matching mechanism $F$, we say that $F$ is decisive if, for every $\overline{p}\in \overline{\mathcal{P}}$, $F(\overline{p})\neq\varnothing$;
resolute if, for every $\overline{p}\in \overline{\mathcal{P}}$, $|F(\overline{p})|=1$; stable if, for every $\overline{p}\in \overline{\mathcal{P}}$, all the elements of $F(\overline{p})$ are stable for $\overline{p}$; Pareto optimal if, for every $\overline{p}\in \overline{\mathcal{P}}$, all the elements of $F(\overline{p})$ are Pareto optimal for $\overline{p}$; $U$-symmetric, where $U\subseteq G^*$, if, for every $\overline{p}\in \overline{\mathcal{P}}$ and $\varphi\in U$, $F(\overline{p}^{\varphi})= F(\overline{p})^\varphi$. As for the model without outside option, we simply say symmetric instead of $G^*$-symmetric.

We stress that, even in the considered framework, the properties of $G$-symmetry (anonymity) and $(G*\setminus G)$-symmetry (gender-fairness) are employed by some authors, specifically to obtain characterization results (see, for instance, Theorem 2 in Miyagawa, 2002 and Theorem 3.1 in Nizamogullari and Ozkal-Sanver, 2014). It is also worth noting that, using the same argument as in the proof of Proposition \ref{iff}, we can deduce that, for every $U\subseteq G^*$, a generalized matching mechanism is $U$-symmetric if and only if it is $\langle U\rangle$-symmetric. Therefore, by Proposition \ref{generato1}, $(G^*\setminus G)$-symmetry is equivalent to symmetry. 

As a consequence of Theorem \ref{impossibilita1}, we are able to prove the following impossibility result, whose proof is in Appendix \ref{main-proofs}.

\begin{theorem}\label{gen-teo}
There exists no resolute, symmetric and Pareto optimal generalized matching mechanism.
\end{theorem}

Of course, Theorem \ref{gen-teo} implies that there exists no resolute, symmetric and stable generalized matching mechanism. On the other hand, the existence of a resolute and symmetric generalized matching mechanism remains an open problem.

\section{Conclusions and future research}

In this paper we focused on balanced markets where $|W|=|M|$ and proved several possibility and impossibility results for matching mechanisms satisfying resoluteness, a given level of symmetry, (i.e. $G^*$-symmetry or $\langle \varphi\rangle$-symmetry), and possibly other desirable properties (i.e. stability, Pareto optimality, weak Pareto optimality and minimal optimality). The two main results of the paper (Theorem \ref{n2} and Theorem \ref{main-mathc}) state that resoluteness and $G^*$-symmetry are consistent with each other if and only if the size of $W$ and $M$ is odd. Moreover, we prove that resoluteness and $G^*$-symmetry are inconsistent with minimal optimality, definitely a very weak condition (Theorem \ref{impossibilita1}); there exist resolute and $\langle\varphi\rangle$-symmetric matching mechanisms that are weak Pareto optimal but none of them can be stable, unless the size of $W$ and $M$ is 2 (Theorem \ref{existence3}, Theorem \ref{existence2}, Theorem \ref{endriss}). Roughly, the conclusions that we can draw from these results could be recap as follows: when we restrict attention to resolute matching mechanisms, the highest level of symmetry, that is $G^*$-symmetry, is incompatible with any sort of optimality condition while weaker symmetry conditions, still involving a comparison between the two groups of agents, are incompatible with strong optimality conditions, like stability, but may be consistent with other interesting optimality conditions, like weak Pareto optimality.

These results are obtained by employing suitable algebraic methods based on group theory, an approach strongly innovative and not yet explored in matching theory. Our paper demonstrates that such an approach is fruitful and able to provide novel and remarkable insight into issues that pertain to fairness.

All these findings constitute a first step in the direction of a complete understanding of the Main Problem described in Section \ref{ResProb}. Several other areas deserve attention for future research.

First of all, one could focus on other forms of fairness which are intermediate between $G^*$-symmetry and
$\langle \varphi\rangle$-symmetry and economically sound, as well. The purpose is to get, in particular, positive results for these intermediate forms of symmetry when $n$ is even. As an example, in a framework where it is natural to identify a partition $\{W_1,W_2\}$ of $W$ and  
a partition $\{M_1,M_2\}$ of $M$ that are economically meaningful and such that $|W_1|=|M_1|$ and $|W_2|=|M_2|$, it could reasonable to focus on $U$-symmetry, where 
\[
U=\{\varphi\in G^*: \varphi(W_1)=W_1 \mbox{ or }\varphi(W_1)=M_1\}.
\] 
In other words, one may focus on those changes in individual identities within the four sets $W_1$, $W_2$, $M_1$ and $M_2$ or the ones across the pair $W_1$ and $M_1$ and the pair $W_2$ and $M_2$. 

Another possible approach for getting possibility results consists in analyzing whether it is possible to identify suitable preference domain restrictions that guarantee the existence of resolute matching mechanisms satisfying strong optimality conditions, like stability, while maintaining a sufficiently high level of symmetry.

Finally, one could ask whether our notion of $U$-symmetry and the
possibility and impossibility results provided in our paper can be
extended to unbalanced markets, namely where the two sets $W$ and $M$ have
different sizes (and where, of course, individuals have the outside option to be unmatched). However, we believe that symmetries involving permutations that change individuals' identities
across the two groups do not make sense in such a framework as the two sets are
substantially different because of their different sizes. On the other hand, symmetries involving only permutations that change individuals' identities within
the two groups (that is, weaker forms of anonymity) may be considered and it seems reasonable to conjecture that any matching mechanisms,  satisfying a symmetry condition of that type, admits a resolute refinement satisfying the same level of symmetry.

\section*{Declarations of competing interest}

The authors have no competing interests to declare that are relevant to the content
of this article.

\section*{Data availability}

No data was used for the research described in the article.

\section*{Acknowledgements}
The authors gratefully thank an Associate Editor and two referees for their valuable comments which contributed to improve the paper.
The authors also wish to express their gratitude to Pablo Spiga for inspiring discussions and suggestions about centralizers of semiregular subgroups. Daniela Bubboloni has been supported by GNSAGA of INdAM (Italy),  by the European Union - Next Generation EU, Missione 4 Componente 1, PRIN 2022-2022PSTWLB - Group Theory and Applications, CUP B53D23009410006, and by local funding from the Universit\`a degli Studi di Firenze. Michele Gori has been supported by local funding from the Universit\`a degli Studi di Firenze. Claudia Meo has been supported by the national project PRIN 2022-2022HLPMKN - Externalities and Fairness in Allocations and Contracts, CUP E53D23006290006.

\appendix

\renewcommand\thesection{\Alph{section}}

\section{Optimality conditions} \label{OptCond}

In this section we first deepen the property of minimal optimality; then, we provide the proof of Proposition \ref{implications} about the relationships among the four optimality conditions introduced in Definition \ref{DefOptCond}.

\begin{proposition}\label{soloMO}
Let $p\in \mathcal{P}$. Then there exists at most one matching that is not minimally optimal for $p$.
\end{proposition}

\begin{proof}
Consider $\mu,\mu'\in\mathcal{M}$ such that $\mu,\mu'\not\in MO(p)$. Then, for every $z\in I$, we have $\mathrm{Rank}_{p(z)}(\mu(z))=n$ and $\mathrm{Rank}_{p(z)}(\mu'(z))=n$. Thus, for every $z\in I$, $\mu(z)=\mu'(z)=\mathrm{Rank}_{p(z)}^{-1}(n)$. Thus, $\mu=\mu'$.
\end{proof}

From the proof of the previous proposition we understand that, given $p\in\mathcal{P}$, we have that $MO(p)=\mathcal{M}\setminus \{\sigma_p\}$, where $\sigma_p:I\to I$ is defined, for every $z\in I$, by $\sigma_p(z)=\mathrm{Rank}_{p(z)}^{-1}(n)$.
Thus, $MO(p)=\mathcal{M}$ if and only if $\sigma_p\notin \mathcal{M}$, that is, it is not bijective or does not satisfies $\sigma_p^2=id_I$ (since surely $\sigma_p(W)\subseteq M$ and $\sigma_p(M)\subseteq W$). For instance, if $p\in \mathcal{P}$ is represented by the table in \eqref{tablep}, then $\sigma_p$ is the matching $(1 6)(25)(34)$ and then $MO(p)\neq\mathcal{M}$.  If instead $p\in \mathcal{P}$ is represented by the table 
\[
\begin{array}{|ccc||ccc|}
\hline
1&2&3&4&5&6\\
\hline
\hline
4&4&6&2&2&3\\
5&5&5&1&1&2\\
6&6&4&3&3&1\\
\hline
\end{array},
\]
we have that $\sigma_p$ is not bijective because $\sigma_p(1)=\sigma_p(2)=6$, and then $MO(p)=\mathcal{M}$. Finally, consider $p\in \mathcal{P}$ represented by the table 
\[
\begin{array}{|ccc||ccc|}
\hline
1&2&3&4&5&6\\
\hline
\hline
4&4&6&3&2&3\\
5&6&5&1&1&2\\
6&5&4&2&3&1\\
\hline
\end{array}.
\]
In this case $\sigma_p$ is bijective. However, $\sigma_p$ fails to satisfy $\sigma_p^2=id_I$ because
$\sigma_p(2)=5$ and $\sigma_p(5)=3\neq 2$. Thus,  $MO(p)=\mathcal{M}$.

Let us introduce some notation that will be used in the proof of Proposition \ref{implications} and also in the sequel.
Let $n\in \mathbb{N}$.  
Given $x\in \mathbb{Z}$, we denote the equivalence class of $x$ modulo $n$ by $[x]_n\coloneq\{x+kn: k\in \mathbb{Z}\}$. Note that, for every  $x\in \mathbb{Z}$, $[x]_n\cap \ldbrack n \rdbrack$ is a singleton whose unique element can be taken as a canonical representative for the set $[x]_n$.

\begin{proof}[Proof of Proposition \ref{implications}]
$(i)$ Let $\mu$ be stable for $p$. Assume by contradiction that $\mu$ is not  Pareto optimal for $p$. Then, there exists $\mu'\in\mathcal{M}$ such that, for every $z \in I$, $\mu'(z) \succeq_{p(z)} \mu(z)$ and there exists $z^* \in I$ such that $\mu'(z^*) \succ_{p(z^*)} \mu(z^*)$. Consider then $z^*$ and set $z' = \mu'(z^*)$.
Then we have $z'=\mu'(z^*) \succ_{p(z^*)} \mu(z^*)$ and $z^*=\mu'(z') \succeq_{p(z')} \mu(z')$. Note that $z^*\neq \mu(z')$. Indeed, assume the contrary. Then
 $z'=\mu(z^*)$ and thus $z'\succ_{p(z^*)} z'$, a contradiction.
As a consequence, we have $z^*=\mu'(z') \succ_{p(z')} \mu(z')$ and hence
 the pair $(z^*,z')$ blocks $\mu$ according to $p$. Thus, $\mu$ is not stable, a contradiction.

$(ii)$ Trivial.

$(iii)$ By Proposition \ref{soloMO}, we have that there exists at most one matching not minimally optimal for $p$. If there is none, then  trivially $WPO(p)\subseteq \mathcal {M}=MO(p).$ Assume instead that $MO(p)= \mathcal {M}\setminus\{\mu^*\}$. In order to prove that  $WPO(p)\subseteq MO(p)$, it is enough to show that $\mu^*\notin WPO(p).$
Let $\mu'\in\mathcal{M}$ be the unique matching such that, for every $y\in M$, $\mu'(y)$ is the unique woman in the set $[\mu^*(y)+1]_n$.
 Observe that, since $n\geq 2$, for every $z\in I$, we have $\mu'(z)\neq\mu^*(z)$ and therefore
$\mathrm{Rank}_{p(z)}(\mu'(z))<n$. As a consequence, for every $z\in I$, $\mu'(z) \succ_{p(z)} \mu^*(z)$. Thus  $\mu^*$ is not weakly Pareto optimal for $p$.
\end{proof}

\section{Group theory in our setup} \label{GT}

\subsection{Preliminary results}\label{GT1}

We start with some additional basic notion from group theory.

Let $H$ be a finite group. Lagrange Theorem  states that the size of a subgroup of $H$ divides the size of $H$ (see Theorem 1.26 in  Milne, 2021). Given $h\in H$, the order of $h$ is defined by $|h|\coloneq|\langle h\rangle|$ and we have that:   
$\langle h\rangle=\{h^j:j\in\mathbb{N}\}$; $|h|$ is the minimum $j\in\mathbb{N}$ such that $h^j=1$; $h^b=1$ for some $b\in \mathbb{N}$ if and only if $|h|$ divides $b$; by Lagrange Theorem, $|h|$ divides $|H|$. Cauchy Theorem partially reverses Lagrange Theorem assuring that if $r$ is a prime dividing $|H|$, then there exists in $H$ an element of order $r$ (see Theorem 4.13 in Milne, 2021).
We say that $H$ is cyclic if there exists $h\in H$ such that $H=\langle h\rangle$. In that case, it is customary to denote the group by $C_m$, where $m=|H|$.

Given $K\le H$, we say that $K$ is a normal subgroup of $H$ if, for every $h\in H$, $K^h=K$ (see Chapter $1$ in Milne, 2021). For the purposes of our paper it is important to emphasize that $G$ is a normal subgroup of $G^*$. Indeed, let $\sigma\in G$ and $\varphi\in G^*$. Then all the cycles in which $\sigma$ splits have support included in $W$ or in $M$ and thus, by \eqref{conj-split} and \eqref{conj-cyc}, the same happens to $\sigma^{\varphi}$. As a consequence, we have $\sigma^{\varphi}\in G$, and hence,  recalling that conjugation establishes a bijection, $G^{\varphi}=G.$ 

Given a finite group $K$, a homomorphism from $H$ to $K$ is a function $f:H\rightarrow K$ such that, for every $h_1,h_2\in H$, $f(h_1h_2)=f(h_1)f(h_2)$, where in both groups the operation are denoted by juxtaposition. 
A homomorphism $f$ from $H$ to $K$ is called (group) isomorphism from $H$ to $K$ if it is bijective. 
The groups $H$ and $K$ are called isomorphic if there exists an isomorphism from $H$ to $K$. In such a case, we write $H\cong K$. 

An action of $H$ on a nonempty finite set $X$ is a homomorphism from $H$ to $\mathrm{Sym}(X)$. Assume that $H$ acts on $X$ through the action $f$. For every $h\in H$ and $x\in X$, we usually write $x^h$ instead of $f(h)(x)$.  In particular, the action is not explicitly mentioned.
Given $x\in X$, the $H$-orbit of $x$ is the  set $x^H \coloneq\{x^h\in X:h\in H\}$; the $H$-stabilizer of $x$ is the subgroup of $H$ defined by $\mathrm{Stab}_H(x) \coloneq\{h\in H: x^h=x\}$. A well-known result establishes that, for every $x\in X$,
\[
|x^H|=\frac{|H|}{|\mathrm{Stab}_H(x)|},
\]
and so, in particular, $|x^H|$ divides $|H|$ (see Chapter $4$ in Milne, 2021).
The set $X^H \coloneq\{x^H\subseteq X:x\in X\}$ is a partition\footnote{A partition of a nonempty set $X$ is a set of nonempty pairwise disjoint subsets of $X$ whose union is $X$.} of $X$ called the set of the $H$-orbits of $X$ or, alternatively, the set of orbits of $H$ on $X$.  The reference to $X$ is usually omitted when $X$ is clear from the context.

Consider a nonempty finite set $X$.  
Given $H\le \mathrm{Sym}(X)$, we can consider the natural action $f$ of $H$ on $X$ defined, for every $\varphi\in H$, by $f(\varphi)=\varphi$. Thus, with respect to that action, for every $x\in X$, we have $f(\varphi)(x)=\varphi(x)$, $x^H=\{\varphi(x)\in X:\varphi\in H\}$ and $\mathrm{Stab}_H(x) =\{\varphi\in H: \varphi(x)=x\}$.

Let $\varphi\in \mathrm{Sym}(X)$.  Since $\langle \varphi\rangle\le\mathrm{Sym}(X)$, we can consider the natural action of $\langle \varphi\rangle$ on $X$. Given $x\in X$, we have that $x^{\langle \varphi\rangle}=\{\varphi^j(x)\in X:j\in\mathbb{N}\}$ and $|x^{\langle \varphi\rangle}|=1$ if and only if $x$ is a fixed point for $\varphi$, that is $\varphi(x)=x$. The set of the $\langle \varphi\rangle$-orbits of $X$, that is the set $X^{\langle \varphi\rangle}=\{x^{\langle \varphi\rangle}\subseteq X: x\in X\}$, is also called the set of the orbits of $\varphi$ on $X$. Let $r=|X^{\langle \varphi\rangle}|\in\mathbb{N}$ be the number of the distinct orbits of $\varphi$ on $X$ and let $x_1,\dots, x_r\in X$ be such that $X^{\langle \varphi\rangle}=\{x_j^{\langle \varphi\rangle}:j\in \ldbrack r \rdbrack\}$. The elements $x_1,\dots, x_r$ are called representatives of the orbits of $\varphi$ on $X$. For every $j\in \ldbrack r\rdbrack$, set $b_j\coloneq |x_j^{\langle \varphi\rangle}|$. The type of $\varphi$ is defined by the unordered list $T_{\varphi}\coloneq[b_1,\dots,b_r]$, whose elements are called the parts of $T_\varphi$. We have that $\sum_{j=1}^rb_j=|X|$ and $|\varphi|=\mathrm{lcm}(b_1,\dots,b_r)$, where lcm denotes the least common multiple. 
Let $J=\{j\in \ldbrack r\rdbrack: b_j\ge 2\}$. If $J=\varnothing$, then $\varphi=id_X$; if $J\neq\varnothing$, then $\varphi\neq id_X$  and $\varphi$ decomposes into the $|J|$ disjoint cycles $\varphi_1, \dots, \varphi_{|J|}$, where, for every $j\in J$, $\varphi_j$ is the $b_j$-cycle given by $\varphi_j\coloneq (x_j\ \varphi(x_j) \cdots \varphi^{b_j-1}(x_j))$. In particular, for every $j\in  \ldbrack r\rdbrack$ and $x\in x_j^{\langle \varphi\rangle}$, $\varphi^{b_j}(x)=x$. 
When we need to emphasize the distinct parts $b_1,\dots, b_s$ of $T_{\varphi}$, the notation $T_{\varphi}=[b_1^{(r_1)},\dots, b_s^{(r_s)}]$ will be used, where, for every $i\in\ldbrack s\rdbrack$, $r_i$ counts how many times $b_i$ appears in $T_\varphi$. As an example, if $T_{\varphi}=[1,1,2,2,3,4,4,4]$, then we can write $T_{\varphi}=\left[1^{(2)},2^{(2)},3,4^{(3)}\right]$.
It is well known that, given $\varphi_1, \varphi_2\in \mathrm{Sym}(X)$, we have $T_{\varphi_1}=T_{\varphi_2}$ if and only if $\varphi_1$ and $\varphi_2$ are conjugate in $\mathrm{Sym}(X)$ (see Proposition 4.30 in Milne, 2021).

Given $S\subseteq \mathrm{Sym}(X)$ and $\varphi\in \mathrm{Sym}(X)$, we set
\[
\varphi S\coloneq\left\{\rho \in  \mathrm{Sym}(X): \exists\sigma \in S \mbox{ such that }\rho=\varphi\sigma\right\},
\]
\[
S\varphi \coloneq\left\{\rho \in  \mathrm{Sym}(X): \exists\sigma \in S \mbox{ such that }\rho=\sigma\varphi\right\}.
\]
Of course, $id_X\, S=S$ and $S\, id_X=S$.
Moreover, for every $\varphi_1,\varphi_2\in \mathrm{Sym}(X)$, we have that $(\varphi_1\varphi_2)S=\varphi_1(\varphi_2 S)$,
$S(\varphi_1\varphi_2)=(S \varphi_1)\varphi_2$, and $(\varphi_1 S)\varphi_2=\varphi_1(S\varphi_2)$. In particular,
given $\varphi_1,\varphi_2,\varphi_3,\varphi_4\in \mathrm{Sym}(X)$, the writing $\varphi_1\varphi_2S\varphi_3\varphi_4$ without brackets is unambiguous: in any way we perform the operations we always get the same set. Note that if $S\subseteq T$ and $\varphi\in \mathrm{Sym}(X)$, then $\varphi S\subseteq \varphi T$ as well as $S\varphi\subseteq T\varphi.$

\subsection{Basic applications and the action on $\mathcal{P}$}\label{GT2}

We collect in this section some fundamental points where group theory comes at play in our model.
Proposition \ref{car-match} establishes some useful properties of matchings. Proposition \ref{action-mat} shows that a subgroup $U$ of $G^*$ acts on the set of preference profiles $\mathcal{P}$, a result that enables to exploit many general facts from group theory. We propose then some consequences of the aforementioned action: the proof of Proposition \ref{iff}, which states that dealing with symmetry for subsets of $G^*$ is equivalent to dealing with symmetry for subgroups of $G^*$, and Propositions \ref{generato1} and \ref{generato2}. Finally, Proposition \ref{STsym} and Proposition \ref{SEsym} are about the different levels of symmetry exhibited by some matching mechanisms studied in the literature.

\begin{proposition} \label{car-match} The following facts hold true:
\begin{itemize}
\item[$(i)$] $\mathcal{M}=\{\mu\in G^*\setminus G: |\mu|=2\}\neq \varnothing;$
\item[$(ii)$] $\mathcal{M}=\{\mu\in G^*\setminus G: T_{\mu}=[2^{(n)}]\};$
\item[$(iii)$] for every $\varphi\in G^*$,  $\mathcal{M}^{\varphi}= \mathcal{M}.$
\end{itemize}
\end{proposition}

\begin{proof}
 $(i)$ Let $\mu\in \mathcal{M}$. Then, by definition, $\mu$ leaves the partition $\{W,M\}$ of $I$ invariant and swaps $W$ and $M$; hence $\mu\in G^*\setminus G$. In particular, $\mu\neq id_I$. Moreover, the fact that $\mu(\mu(z))=z$ for all $z\in I$ implies that $|\mu|=2.$ The other inclusion is obvious. The nonemptyness of $\mathcal{M}$ is immediate since $\mu= (1\ n+1)\cdots (n\ 2n)$ is a matching.

$(ii)$ By $(i)$, we just need to show that $\mu\in G^*\setminus G$ has order $2$ if and only if $T_{\mu}=[2^{(n)}].$ If $T_{\mu}=[2^{(n)}]$, then obviously $|\mu|=2$. Let  $\mu\in G^*\setminus G$ with $|\mu|=2$. Then $T_{\mu}=[2^{(k)}, 1^{(s)}]$ with $2k+s=2|W|$ and $s\geq 0.$ Since $\mu$ takes $W$ into $M$ and $M$ into $W$, every orbit of $\mu$ contains at least an element in $W$  and an element in $M$.
As a consequence, no part of $T_{\mu}$ can be smaller than $2$ and thus $s=0.$

$(iii)$ We first show that, for every $\varphi\in G^*$, $\mathcal{M}^\varphi\subseteq \mathcal{M}$.
Let $\varphi\in G^*$. Consider then $\mu\in\mathcal{M}$ and prove that $\mu^{\varphi}\in\mathcal{M}$. By $(i)$, we have $\mu\in G^*\setminus G$ and  $|\mu|=2.$ Since $G^*$ is a group, $\mu^\varphi\in G^*$. Assume that $\mu^\varphi\in G.$ Since $G$ is normal in $G^*$, we deduce that $\mu=(\mu^\varphi)^{\varphi^{-1}} \in G$, a contradiction. Thus $\mu^\varphi\in G^*\setminus G$. Moreover, we have $|\mu^\varphi|=|\mu|$ as $\mu$ and $\mu^\varphi$ are conjugate. Thus, by $(i)$, we get $\mu^\varphi\in \mathcal{M}$. 

We complete the proof showing that, for every $\varphi\in G^*$, $\mathcal{M}\subseteq \mathcal{M}^\varphi$.
Let  $\varphi\in G^*$. Since $\varphi^{-1}\in G^*$, by the first part of the proof, we get
$\mathcal{M}^{\varphi^{-1}}\subseteq \mathcal{M}$ and then
$\mathcal{M}=(\mathcal{M}^{\varphi^{-1}})^\varphi\subseteq \mathcal{M}^\varphi$.
\end{proof}

\begin{proposition}\label{action-mat}
Let $U\le G^*$. Then, the function $f:U\to \mathrm{Sym}(\mathcal{P})$ defined, for every $\varphi\in U$, by
\[
f(\varphi):\mathcal{P}\to \mathcal{P},\quad\quad p\mapsto f(\varphi)(p)= p^{\varphi}
\]
is well defined and it is an action of $U$ on the set $\mathcal{P}$.
\end{proposition}

\begin{proof}
As a preliminary step, let us prove that, for every $p\in\mathcal{P}$ and $\varphi_1,\varphi_2\in G^*$, we have that 
\begin{equation}\label{action-e}
\left(p^{\varphi_2}\right)^{\varphi_1}=p^{\varphi_1\varphi_2}.
\end{equation}
Consider then $p\in\mathcal{P}$ and $\varphi_1,\varphi_2\in G^*$.
We have to show that, for every $z\in I$, $\left(p^{\varphi_2}\right)^{\varphi_1}(z)=p^{\varphi_1\varphi_2}(z)$.
Consider $z\in I$. Then we have
\[
\left(p^{\varphi_2}\right)^{\varphi_1}(z)=\varphi_1 p^{\varphi_2}(\varphi_1^{-1}(z))=
\varphi_1\left(\varphi_2 p(\varphi_2^{-1}(\varphi_1^{-1}(z)))\right)
=\left(\varphi_1 \varphi_2\right)p((\varphi_1\varphi_2)^{-1}(z))=p^{\varphi_1\varphi_2}(z),
\]
as desired.

Let us fix now $\varphi\in U$ and prove that $f(\varphi)\in \mathrm{Sym}(\mathcal{P})$. Since $\mathcal{P}$ is finite, it is enough to show that $f(\varphi)$ is surjective. Consider $p\in \mathcal{P}$. By \eqref{action-e},
\[
f(\varphi)\left(p^{\varphi^{-1}}\right)=\left(p^{\varphi^{-1}}\right)^{\varphi}=p^{id_{I}}=p.
\]
Since, by \eqref{action-e}, we also have that, for every $\varphi_1,\varphi_2\in U$, $f(\varphi_1\varphi_2)=f(\varphi_1)f(\varphi_2)$,
we have that $f$ is a group homomorphism and so we conclude that $f$ is an action of $U$ on $\mathcal{P}$.
\end{proof}

\begin{proof}[Proof of Proposition \ref{iff}]
Assume that $F$ is $\langle U \rangle$-symmetric. Then,
for every $p\in \mathcal{P}$ and $\varphi\in \langle U\rangle$, $F(p^{\varphi})=F(p)^\varphi$. In particular, for every $p\in \mathcal{P}$ and $\varphi\in U$, $F(p^{\varphi})=F(p)^\varphi$. Thus, we conclude that $F$ is $U$-symmetric.

Assume now that $F$ is $U$-symmetric. Then, for every $p\in \mathcal{P}$ and $\varphi\in U$, $F(p^{\varphi})=F(p)^\varphi$. It remains to prove that, for every $p\in \mathcal{P}$ and $\varphi\in \langle U\rangle$, $F(p^{\varphi})=F(p)^\varphi$. Let us set $U_1=U$ and, for every $k\in\mathbb{N}$, let
\[
U_k\coloneq\{\varphi\in G^*:\exists \sigma_1,\ldots,\sigma_k\in U\mbox{ such that }\varphi=\sigma_1\cdots \sigma_k\}.
\]
For every $k\in\mathbb{N}$, consider the following statement, denoted by $S(k)$,
\[
\mbox{for every $p\in \mathcal{P}$ and $\varphi\in U_k$, we have that  $F(p^{\varphi})= F(p)^\varphi$.}
\]
We claim that, for every $k\in\mathbb{N}$, $S(k)$ is true. We prove that fact by induction. The truth of $S(1)$ is immediate. Assume that $S(k)$ is true and prove $S(k+1)$ is true. Let $p\in \mathcal{P}$ and $\varphi\in U_{k+1}$ and prove that  $F(p^{\varphi})=F(p)^\varphi$. Since $\varphi\in U_{k+1}$, we can find $\psi\in U_k$ and $\sigma \in U=U_1$ such that $\varphi=\psi\sigma$. Thus, by \eqref{action-e}, we have that $p^\varphi=p^{\psi\sigma}=
\left(p^{\sigma}\right)^{\psi}$. Thus,
\[
F(p^\varphi)=F(p^{\psi\sigma})=F(\left(p^{\sigma}\right)^{\psi})=F(p^\sigma)^\psi=
(F(p)^\sigma)^\psi=F(p)^{\psi\sigma}=F(p)^\varphi,
\]
where the third equality follows from the fact that $S(k)$ is true, the forth one from the fact that $S(1)$ is true and the fifth one from \eqref{conprop}.

By the claim and the fact that $ \langle U \rangle=\bigcup_{k\in\mathbb{N}}U_k$, we then deduce that, for every $p\in \mathcal{P}$ and $\varphi\in \langle U\rangle$, $F(p^{\varphi})= F(p)^\varphi$, as desired.
\end{proof}

\begin{proposition}\label{generato1}
$\langle G^*\setminus G\rangle=G^*$.
\end{proposition}

\begin{proof} Let $U\coloneq\langle G^*\setminus G\rangle$. We want to show that $U=G^*$. Since $U\supseteq  G^*\setminus G$, it is enough to show that $U\supseteq G$. Consider then $\psi \in G$ and prove that $\psi \in U$. Pick $\varphi\in G^*\setminus G$. Using the fact that $G$ is normal in $G^*$, we get $G=G^{\varphi}$. Then there exists $\nu\in G$ such that $\psi=\nu^{\varphi}$. Hence, we have $\psi=\varphi (\nu \varphi^{-1}).$ We cannot have $\nu \varphi^{-1}\in G$, because that would imply $\varphi \in G.$ Thus, we have expressed $\psi$ as the product of two elements in $G^*\setminus G,$ and hence $\psi\in U$.
\end{proof}

\begin{proposition}\label{generato2}
Let $\varphi\in\mathcal{M}$. Then $\langle G\cup\{\varphi\}\rangle=G^*$.
\end{proposition}

\begin{proof} Let $U\coloneq \langle G\cup\{\varphi\}\rangle$. We want to show that $U=G^*$. Since $U\supseteq  G\cup\{\varphi\}$ and $\varphi\in G^*\setminus G$, we have that $U\ge G$ and $U\neq G$. By Lagrange Theorem, we then have that $|G|$ properly divides $|U|$. Then $|U|=|G|k$, for some $k\in \mathbb{N},$ with $k\geq 2$. On the other hand, using again Lagrange Theorem, we also have that $|U|$ divides $|G^*|=2|G|$. As a consequence $|U|=2|G|=|G^*|$ and hence $U=G^*$.
\end{proof}

\begin{proposition}\label{STsym}
$GS_w$ and $GS_m$ are $G$-symmetric but they are not symmetric. $GS$, $ST$, $PO$, $WPO$, $MO$ and $TO$ are symmetric.
\end{proposition}

\begin{proof}
The $G$-symmetry of $GS_w$ and $GS_m$ follows from the nature of the algorithm in Gale and Shapley (1962). Let us prove that $GS_w$ and $GS_w$ are not symmetric.
Indeed, let $W=\ldbrack n\rdbrack$, $M=\{y_1,\dots, y_n\}$,
and consider the following preference profile $p$:
\[
p=
\begin{array}{|ccccc||ccccc|}
\hline
1 & 2 & (3)  & \ldots & (n) & y_1 & y_2 & (y_3) & \ldots & (y_n) \\
\hline
\hline
y_1 & y_2 & (y_3 )& \ldots & (y_n) & 2 & 1 & (3) &  \ldots & (n) \\
\vdots & \vdots & \vdots &\vdots & \vdots & \vdots  & \vdots & \vdots & \vdots  & \vdots     \\
\hline
\end{array}
\]
where the vertical dots mean that the missing entries of each linear order can be listed in any possible order.
Consider the permutation $\varphi=(1 \, y_1)(2 \,  y_2)\cdots (n \,  y_n) \in G^*$ that switches woman $x$ and man $y_x$ for each $x\in\ldbrack n\rdbrack$.
Let $\gamma^*=id_I$ if $n= 2$ while $\gamma^*=(3\,y_3)\cdots(n\,y_n)\in \mathrm{Sym}(I)$ if $n\ge 3$. Note that, for every $n\geq 2$, we have $\varphi=(1 \, y_1)(2 \,  y_2)\gamma^*$.
It is immediate to check that $GS_w(p)=\{\varphi\}$ and $GS_m(p)=\{(1 \, y_2) (2\,  y_1) \gamma^*\}$, and then 
$GS_w(p)=GS_w (p)^\varphi$ and $GS_m(p)=GS_m (p)^\varphi$. We also have that
\begin{eqnarray*}
p^{\varphi}=
& \begin{array}{|ccccc||ccccc|}
\hline
y_1 & y_2 & (y_3) &  \ldots & (y_n) & 1 & 2 & (3) & \ldots & (n) \\
\hline
\hline
1 & 2 & (3)&  \ldots & (n) & y_2 & y_1 & (y_3) &  \ldots & (y_n) \\
\vdots & \vdots &\vdots & \vdots & \vdots & \vdots & \vdots & \vdots & \vdots& \vdots   \\
\hline
\end{array}
\\
=&\begin{array}{|ccccc||ccccc|}
\hline
1 & 2 & (3) &  \ldots & (n) & y_1 & y_2 & (y_3) & \ldots & (y_n) \\
\hline
\hline
y_2 & y_1 & (y_3 )& \ldots & (y_n) & 1 & 2 & (3 )& \ldots & (n )\\
\vdots & \vdots & \vdots & \vdots & \vdots & \vdots  & \vdots & \vdots & \vdots & \vdots   \\
\hline
\end{array}.
\end{eqnarray*}
Thus,  $GS_w(p^{\varphi})=\{(1\, y_2)(2\, y_1)\gamma^*\}\neq GS_w (p)^\varphi$ and
 $GS_m(p^{\varphi})=\{\varphi\} \neq GS_m (p)^\varphi$.
As a consequence, $GS_w$ and $GS_m$ fail to be symmetric.

In order to prove that $GS$ is symmetric, we need to show that, for every $p \in \mathcal{P}$ and $\varphi \in G^*$, $GS(p^\varphi)=GS(p)^\varphi$. If $\varphi \in G$, that is a consequence of the $G$-symmetry of $GS_w$ and $GS_m$. Indeed, $GS(p^{\varphi})= \{GS_w(p^{\varphi}), GS_m(p^{\varphi})\}=\{GS_w(p)^\varphi,GS_m(p)^\varphi \}=GS(p)^\varphi$. If $\varphi \in G^*\setminus G$, the proof has been provided in Proposition 3.3 by \"{O}zkal-Sanver (2004).

The fact that $ST$ is symmetric follows immediately from the following property: if $\mu \in \mathcal{M}$, $p\in \mathcal{P}$ and $\varphi\in G^*$, the pair $(x,y) \in W \times M$ blocks $\mu$ according to $p$ if and only if the pair $(\varphi(x),\varphi(y))$ blocks  $\mu^\varphi$  according to  $p^{\varphi}$. We remark that $(\varphi(x),\varphi(y)) \in W \times M$ if $\varphi \in G$, while  $(\varphi(x),\varphi(y)) \in M \times W$ if $\varphi \in G^*\setminus G$.
Let us prove the stated property. By definition, $(x,y) \in W\times M$ blocks $\mu$ according to  $p$ if
\[
 y \succ_{p(x )}\mu(x) \hbox{ and } x \succ_{p(y)} \mu(y).
\]
By \eqref{action-w2}, that is equivalent to
\[
\varphi(y) \succ_{p^{\varphi}(\varphi(x))} \varphi (\mu(x)) \hbox{ and } \varphi(x) \succ_{p^{\varphi}(\varphi(y))} \varphi (\mu(y)),
\]
which, in turn, is equivalent to
\[
 \varphi(y) \succ_{p^\varphi (\varphi(x))} \mu^\varphi(\varphi(x)) \hbox{ and } \varphi(x) \succ_{p^\varphi (\varphi(y))} \mu^\varphi(\varphi(y)),
\]
that is,  $(\varphi(x),\varphi(y))$ blocks  the matching $\mu^\varphi$ according to $p^{\varphi}$.

The fact that $WPO$ is symmetric follows immediately from the following property: for every $\mu \in \mathcal{M}$ and $\varphi \in G^*$,
$\mu$  is not weakly Pareto optimal for $p$ if and only if $\mu^\varphi$ is not weakly Pareto optimal for $p^{\varphi}$.
In order to prove the aforementioned property, consider $\mu \in \mathcal{M}$ and $\varphi \in G^*$ and note that, by definition, $\mu$ is not weakly Pareto optimal for $p$ if and only if there exists a matching $\mu' \in \mathcal{M}$ such that, for every $z\in I$, $\mu'(z) \succ_{p(z)}\mu (z)$. For every $z\in I$, we have that $\mu'(z) \succ_{p(z)}\mu (z)$ is equivalent to $\varphi(\mu'(z))  \succ_{p^{\varphi}(\varphi(z))} \varphi(\mu(z))$ that in turn is equivalent to $(\mu')^\varphi(\varphi(z)) \succ_{p^{\varphi}(\varphi(z))}  \mu^\varphi(\varphi(z))$. Since that holds true for each $z\in I$ and since $(\mu')^\varphi\in \mathcal{M}$, that is equivalent to state that the matching $\mu ^\varphi$ is not weakly Pareto optimal for $p^{\varphi}$.

The proof that $PO$ is symmetric is similar to the proof that $WPO$ is symmetric and then omitted

Let us prove that $MO$ is symmetric. This is an immediate consequence of the following property: for every $\mu \in \mathcal{M}$ and $\varphi \in G^*,$ $\mu$ is  not minimally optimal for $p$ if and only if $\mu^\varphi$ is  not minimally optimal for $p^{\varphi}$.
That property can be proved by noticing that $\mu \not\in MO(p)$ is equivalent to $\hbox{Rank}_{p(z)}\mu(z)=n$ for all $z \in I$, which in turn is equivalent to $\hbox{Rank}_{p^{\varphi}(\varphi(z))}(\mu^\varphi)(\varphi(z))=n$ for all $z\in I$, which is equivalent to $\mu^\varphi \not\in MO(p^{\varphi})$.

Finally,  $TO$ is symmetric since, for every $p \in \mathcal{P}$ and $\varphi \in G^*$, we have that $TO(p)=TO(p^{\varphi})= \mathcal{M}$ and, by Proposition \ref{car-match}$(iii)$, $\mathcal{M}^\varphi= \mathcal{M}$.
\end{proof}

\begin{proposition}\label{SEsym}
$SE$ and $ES$ are symmetric.
\end{proposition}
\begin{proof}
Let us prove that $SE$ is symmetric.
Observe first that, immediately from the definition of $p^{\varphi}$, for every $\varphi \in G^*$, $p \in \mathcal{P}$ and $z,u\in I$ of different gender, we have that
\begin{equation}\label{fatto1}
 \hbox{Rank}_{p^{\varphi}(\varphi(z))}\varphi(u)=\hbox{Rank}_{p(z)}u.
\end{equation}
We next show that, for every $\varphi \in G^*$, $p \in \mathcal{P}$ and $\mu \in \mathcal{M}$, we have
\begin{equation}\label{fatto2}
\delta(p^{\varphi}, \mu^{\varphi})= \delta(p, \mu).
\end{equation}
Indeed, using \eqref{fatto1}, we have
\begin{eqnarray*}
\delta(p^{\varphi}, \mu^{\varphi})&=&\left|\sum_{x \in W}\mathrm{Rank}_{p^{\varphi}(x)}(\mu^{\varphi}(x))-\sum_{y \in M}\mathrm{Rank}_{p^{\varphi}(y)}(\mu^{\varphi}(y))\right|\\
&=&\left|\sum_{x \in W}\mathrm{Rank}_{p^{\varphi}(\varphi(x))}(\mu^{\varphi}(\varphi(x)))-\sum_{y \in M}\mathrm{Rank}_{p^{\varphi}(\varphi(y))}(\mu^{\varphi}(\varphi(y)))\right|
\\
&=&\left|\sum_{x \in W}\mathrm{Rank}_{p^{\varphi}(\varphi(x))}(\varphi\mu(x))-\sum_{y \in M}\mathrm{Rank}_{p^{\varphi}(\varphi(y))}(\varphi\mu(y))\right|\\
&=&\left|\sum_{x \in W}\mathrm{Rank}_{p(x)}(\mu(x))-\sum_{y \in M}\mathrm{Rank}_{p(y)}(\mu(y))\right|=\delta(p, \mu).
\end{eqnarray*}
As a consequence, using the fact that $ST$ is symmetric and applying \eqref{fatto2}, we have that, for every $p\in\mathcal{P}$
and $\varphi \in G^*$,
\begin{eqnarray*}
SE(p^{\varphi}) &=&  \underset{\mu \in ST(p^{\varphi})}{\mathrm{arg\,min}}\; \delta(p^{\varphi},\mu)=
\underset{\mu \in ST(p)^\varphi}{\mathrm{arg\,min}}\; \delta(p^{\varphi},\mu)=\left(\underset{ \sigma \in ST(p)}{\mathrm{arg\,min}}\; \delta(p^{\varphi},\sigma^{\varphi})\right)^\varphi\\
&=&\left(\underset{ \sigma \in ST(p)}{\mathrm{arg\,min}}\; \delta(p,\sigma)\right)^\varphi=SE(p)^\varphi.
\end{eqnarray*}
Thus $SE$ is symmetric.

The proof that $ES$ is symmetric follows the same line of reasoning and  hence, for the sake of brevity, is omitted.
\end{proof}

\section{Proof of Theorem \ref{F-ref-consistent} }\label{Appendix-orbits}

Thanks to the fact that the function $f$ defined in Proposition \ref{action-mat} is an action of $G^*$ on $\mathcal{P}$, we can transfer to the framework of matching theory notation and results concerning the action of a group on a set.
Let $U\le G^*$. For every $p\in \mathcal{P}$, the $U$-orbit of $p$ in $\mathcal{P}$ is $p^U=\{p^\varphi\in \mathcal{P}: \varphi\in U\}$. We denote by $R$ the size of the set $\mathcal{P}^U=\{p^U:p\in \mathcal{P}\}$ of the $U$-orbits of $\mathcal{P}$.\footnote{The number $R$ depends of course on $U$ and $\mathcal{P}$. We can avoid more precise notation since $U$ and $\mathcal{P}$ are always clear from the context.} A system of representatives of the $U$-orbits of $\mathcal{P}$ is any $(p_j)_{j=1}^{R}\in \mathcal{P}^{R}$ such that $\mathcal{P}^U=\{p_j^{U} : j\in \ldbrack R \rdbrack\}$.

Before proving Theorem \ref{F-ref-consistent}, we propose the following preliminary result. 

\begin{theorem}\label{existence}
Let $U\le G^*$, $(p_j)_{j=1}^R\in \mathcal{P}^R$ be a system of representatives for the
$U$-orbits of $\mathcal{P}$, and $(\mu_j)_{j=1}^R$ be a sequence of matchings such that, for every $j\in \ldbrack R \rdbrack$, $\mu_j\in C^U(p_j)$.
Then, there exists a unique resolute and $U$-symmetric matching mechanism $H$ such that, for every $j\in \ldbrack R \rdbrack$, $H(p_j)=\{\mu_j\}$.
\end{theorem}

\begin{proof} Observe first that, given $p\in \mathcal{P}$, there exists $j\in \ldbrack R \rdbrack$ and $\varphi_j\in U$ such that $p=p_j^{\varphi_j}$.
Consider then the matching mechanism $H$ defined, for every $p\in \mathcal{P}$, by $H(p)\coloneq\{\mu_j^{\varphi_j}\}$,
where $j\in \ldbrack R \rdbrack$ and $\varphi_j\in U$ are such that $p=p_j^{\varphi_j}$.

We first prove that $H$ is well defined. Let $p\in \mathcal{P}$ and assume that
$p=p_{j_1}^{\varphi_{1}}$ and $p=p_{j_2}^{\varphi_{2}}$,
where $j_1,j_2\in \ldbrack R \rdbrack$ and $\varphi_{1},\varphi_{2}\in U$.
We have to prove that $\mu_{j_1}^{\varphi_{1}}=\mu_{j_2}^{\varphi_{2}}$.
Note that, since  $(p_j)_{j=1}^R\in \mathcal{P}^R$ is a system of representatives for the
$U$-orbits in $\mathcal{P}$, it must be $j_1=j_2$; therefore $\mu_{j_1}=\mu_{j_2}$ and $p_{j_1}=p_{j_2}$.
We now have that $\mu_{j_1}^{\varphi_{1}}=\mu_{j_1}^{\varphi_{2}}$
is equivalent to $(\mu_{j_1}^{\varphi_{1}})^{\varphi_{2}^{-1}} = \mu_{j_1}$ that, in turn, is equivalent to $\mu_{j_1}^{\varphi_{2}^{-1}\varphi_{1}} = \mu_{j_1}$.
In order to prove the equality $\mu_{j_1} ^{\varphi_{2}^{-1}\varphi_{1}} = \mu_{j_1}$, observe that $\varphi_{2}^{-1}\varphi_{1}\in \mathrm{Stab}_U(p_{j_1})$. Indeed, from
$p=p_{j_1}^{\varphi_{1}}$ and $p=p_{j_1}^{\varphi_{2}}$,
we get
$p_{j_1}^{\varphi_{1}}=p_{j_1}^{\varphi_{2}}$
and then, using the action, we have
$
\left(p_{j_1}^{\varphi_{1}}\right)^{\varphi_{2}^{-1}}=p_{j_1},
$
that is,
$
p_{j_1}^{\varphi_{2}^{-1}\varphi_{1}}=p_{j_1}.
$
Since $\mu_{j_1}\in C^U(p_{j_1})$, we can then conclude that
$\mu_{j_1}^{\varphi_{2}^{-1}\varphi_{1}} = \mu_{j_1}$,
as desired. Thus, $H$ is well defined.

Of course, $H$ is resolute and since, for every $j\in \ldbrack R \rdbrack$, $p_j=p_j^{id_I}$, we have that $H(p_j)=\{\mu_j\}$.

Let us prove that $H$ is $U$-symmetric.
Consider then $p\in \mathcal{P}$ and $\varphi\in U$ and let us prove that
$H(p^{\varphi})=H(p)^\varphi$.
Let $j\in \ldbrack R \rdbrack$ and $\varphi_j\in U$ be such that $p=p_j^{\varphi_j}$. Thus,
$H(p)=\{\mu_j^{\varphi_j}\}=H(p_j)^{\varphi_j}$ and
\[
H(p^{\varphi})=H\left(\left(p_j^{\varphi_j}\right)^{\varphi}\right)=H\left(p_j^{\varphi\varphi_j}\right)
=H(p_j)^{\varphi\varphi_j}
=(H(p_j)^{\varphi_j})^{\varphi}=H(p)^{\varphi}.
\]

It remains to prove the uniqueness of $H$. Consider a $U$-symmetric and resolute matching mechanism $H'$ such that, for every $j\in \ldbrack R \rdbrack$, $H'(p_j)=\mu_j$. We have to prove that, for every $p\in \mathcal{P}$, $H(p)=H'(p)$. Consider then $p\in \mathcal{P}$. There exist
$j\in \ldbrack R \rdbrack$ and $\varphi_j\in U$ such that $p=p_j^{\varphi_j}$.
Thus, using the $U$-symmetry of $H$ and $H'$ and the fact that $H(p_j)=H'(p_j)=\{\mu_j\}$, we get
\[
H'(p)=H'(p_j^{\varphi_j})=H'(p_j)^{\varphi_j}=\{\mu_j^{\varphi_j}\}=H(p_j)^{\varphi_j}=
H(p_j^{\varphi_j})=H(p),
\]
as desired.
\end{proof}

\begin{proof}[Proof of Theorem \ref{F-ref-consistent}]
Assume that $F$ admits a $U$-symmetric and resolute refinement $H$. Using Proposition \ref{facile1}, we deduce that, for every
$p\in \mathcal{P}$, $H(p)\subseteq F(p)\cap C^U(p)$ and so $F(p)\cap C^U(p)\neq\varnothing$.

Conversely, assume that, for every $p\in \mathcal{P}$, $F(p)\cap C^{U}(p)\neq\varnothing$. Let $(p_j)_{j=1}^R\in \mathcal{P}^R$ be a system of representatives for the $U$-orbits in $\mathcal{P}$ and $(\mu_j)_{j=1}^R$ be a sequence of matchings such that, for every $j\in \ldbrack R \rdbrack$, $\mu_j\in F(p_j)\cap C^U(p_j)\neq\varnothing$.
By Theorem \ref{existence} there exists a unique $U$-symmetric and resolute matching mechanism $H$ on $\mathcal{P}$ such that, for every $j\in \ldbrack R \rdbrack$, $H(p_j)=\{\mu_j\}$.
Let us prove that $H$ is a refinement of $F$. Let $p\in \mathcal{P}$ and consider $j\in\ldbrack R \rdbrack$ and $\varphi_j\in U$ such that $p=p_j^{\varphi_j}$.
Then, using the $U$-symmetry of $H$ and $F$, we get
\[
H(p)=H(p_j^{\varphi_j})=H(p_j)^{\varphi_j}=\{\mu_j^{\varphi_j}\}\subseteq F(p_j)^{\varphi_j}=F(p_j^{\varphi_j})=F(p).
\]
\end{proof}

\begin{corollary}\label{facile3}
Let $U\le G^*$. Then, there exists a resolute and $U$-symmetric matching mechanism
if and only if, for every $p\in \mathcal{P}$, $C^U(p)\neq\varnothing$.
\end{corollary}

\begin{proof}
Simply set $F=TO$ and apply Theorem \ref{F-ref-consistent}.
\end{proof}

\section{Semiregular permutation groups}\label{semi-reg}

\subsection{Semiregular subgroups of $G^*$}\label{sect-semi}

We start this section recalling some definitions from the theory of permutation groups.

Let $X$ be a nonempty and finite set and $S$ be a subgroup of $\mathrm{Sym}(X)$.  $S$ is called semiregular if, for every $x\in X$, we have $\mathrm{Stab}_S(x)=\{id_X\}$. In other words, $S$ is semiregular if the only $\varphi\in S$ fixing at least an element of $X$ is $\varphi=id_X.$\footnote{The concept of semiregularity is widely considered in the literature of permutation groups. Some recent contributions involving semiregular groups are, for instance,  Maynard and Siemons (2002), Bereczky and Mar\`oti (2008), Kov\`acs et al. (2015).} $S$ is called transitive if, for every $x,y\in X$, there exists $s\in S$ such that $s(x)=y$; regular if $S$ is transitive and semiregular. Clearly $S$ is transitive if and only if $S$ admits a unique orbit on $X$. Moreover, if $S$ is transitive, then $|X|$ divides $|S|$ because, once chosen $x\in X$, we have  $|X|=|x^S|=\frac{|S|}{|\mathrm{Stab}_S(x)|}$.

The following proposition describes some general properties of semiregular groups of permutations.\footnote{Even though Proposition \ref{facili-semireg} expresses classic and well-established facts, we could not find a  reference containing the statements necessary for our scopes. So, for clarity of exposition, we have preferred to state and prove what needed.}

\begin{proposition}\label{facili-semireg} Let $S\leq \mathrm{Sym}(X)$ be semiregular. Then, the following facts hold true:
\begin{itemize}
\item[$(i)$] every orbit of $S$ has size $|S|$;
\item[$(ii)$] $|S|$ divides $|X|$ and the number of orbits of $S$ is $\frac{|X|}{|S|}$;
\item[$(iii)$] $S$ is regular if and only if $|S|=|X|$;

\item[$(iv)$] every subgroup of $S$ is semiregular;
\item[$(v)$] if $s_1,s_2\in S$ are such that there exists $x\in X$ with $s_1(x)=s_2(x)$, then $s_1=s_2$.
\end{itemize}
\end{proposition}

\begin{proof}
$(i)$ Pick $x\in X$. Then, $|x^S|=\frac{|S|}{|\mathrm{Stab}_S(x)|}=|S|$.

$(ii)$ We know that $X$ is disjoint union of the orbits of $S$. By $(i)$, each orbit of $S$ has size $|S|$. Thus, we deduce that
$|S|$ divides $|X|$ and that the number of orbits of $S$ is $\frac{|X|}{|S|}$;

$(iii)$ Assume that $S$ is regular. Then, $S$ is transitive and hence we have $x^S=X$ and so $|x^S|=|X|$. By the properties of the orbits, we know then that $|X|$ divides $|S|$. Moreover, by $(i)$, we know that $|S|$ divides $|X|$. We can thus conclude that $|S|=|X|$. Conversely, assume that $|S|=|X|$. Then, every orbit of $S$ on $X$ has size $|X|$. Thus, there is a unique orbit and $S$ is transitive.

$(iv)$ Let $U\leq S$. Then, for every $x\in X$, we have $\mathrm{Stab}_U(x)=U\cap \mathrm{Stab}_S(x)=U\cap \{id_X\}= \{id_X\}$. Thus, $U$ is semiregular.

$(v)$ Let $s_1,s_2\in S$ and $x\in X$ be such that $s_1(x)=s_2(x)$. Thus, $s_2^{-1}s_1(x)=x$, that is, $s_2^{-1}s_1\in \mathrm{Stab}_S(x)= \{id_X\}$. As a consequence, $s_2^{-1}s_1=id_X$, that is, $s_1=s_2$.
\end{proof}

We emphasize the strong consequence of semiregularity expressed by Proposition \ref{facili-semireg}$(v)$: if two permutations of $S$ coincide on a single element of $X$, then they coincide on the whole set $X$.

From here on, we focus on the properties of semiregular subgroups of $G^*$ under the assumption that $n$ is odd.

\begin{proposition}\label{orbiteS1}
Let $n$ be odd and $S\le G$ be semiregular. Then, the following facts hold true:
\begin{itemize}
\item[$(i)$] $|S|$ divides $n$, $|S|$ is odd and the number of orbits of $S$ is even;
\item[$(ii)$] every orbit of $S$ is included in $W$ or is included in $M$;
\item[$(iii)$] there are $\frac{n}{|S|}$ orbits of $S$ included in $W$ and $\frac{n}{|S|}$ orbits of $S$ included in $M$.
\end{itemize}
\end{proposition}
\begin{proof}
Since $S\leq G$, every $\varphi\in S$ takes $W$ in $W$ and $M$ in $M$. Thus, every orbit of $S$ is included in $W$ or in $M$ and this proves $(ii)$. By $(ii)$, we know that $W$ is union of orbits of $S$ and $M$ is union of orbits of $S$. Since $S$ is semiregular, by Proposition \ref{facili-semireg}, we conclude that each orbit of $S$ has size $|S|$. We conclude that $|S|$ divides $|W|=|M|=n$. Thus, there are $\frac{n}{|S|}$ orbits of $S$ included in $W$ and $\frac{n}{|S|}$ orbits of $S$ included in $M$ and this proves $(iii)$. Moreover, since $n$ is odd, we get that $|S|$ is odd. Moreover, by Proposition \ref{facili-semireg}, we know that the number of orbits of $S$ is $\frac{2n}{|S|}$ , which is even because $|S|$ is odd. This proves $(i)$.
\end{proof}

\begin{proposition}\label{orbiteS2} Let $n$ be odd and $S\le G^*$ be semiregular and such that $S\not\le G$. Then, the following facts hold true:
\begin{itemize}
\item[$(i)$] $|S|$ divides $2n$, $|S|$ is even and the number of orbits of $S$ is odd;
\item[$(ii)$] for every $z\in I$, $|z^S\cap W|=|z^S\cap M|=\frac{|S|}{2}$.
\end{itemize}
\end{proposition}
\begin{proof}
Let $z\in I$ and prove that $|z^S\cap W|=|z^S\cap M|$. Since $S\not\leq G$ there exists $s\in S\setminus G\subseteq G^*\setminus G$. Since $s$ takes $W$ in $M$, we can consider the function  $\overline{s}:z^S\cap W\rightarrow z^S\cap M$ defined, for every $x\in z^S\cap W$, by $\overline{s}(x)=s(x)$.
We have that $\overline{s}$ is injective, because if $x_1,x_2\in z^S\cap W$ are such that $\overline{s}(x_1)=\overline{s}(x_2)$, then we have that
$s(x_1)=s(x_2)$ and thus $x_1=x_2$. Let us prove that $\overline{s}$ is surjective. First of all, let us observe that $s^{-1}\in S\setminus G$. Indeed  $s^{-1}\in S$ because $S$ is a group. Moreover, $s^{-1}\notin G$ because, being $G$ a group, $s^{-1}\in G$ implies $s\in G$. Consider now $y\in z^S\cap M$ and note that $s^{-1}(y)\in z^S\cap W$ and $\overline{s}(s^{-1}(y))=ss^{-1}(y)=y$. We conclude then that $\overline{s}$ is a bijection, and so $|z^S\cap W|=|z^S\cap M|$. Using Proposition \ref{facili-semireg}, we obtain that $$|S|=|z^S|=|z^S\cap (W\cup M)|=
|(z^S\cap W)\cup(z^S\cap M)|=|z^S\cap W|+|z^S\cap M|=2|z^S\cap W|.$$ Thus, $|S|$ is even  and $|z^S\cap W|=|z^S\cap M|=\frac{|S|}{2}.$
Finally, by Proposition \ref{facili-semireg}, we know that $|S|$ divides $2n$. Since $n$ is odd and $|S|$ is even, we conclude that the number of orbits of $S$, that is $\frac{2n}{|S|}$, is odd. This proves $(i)$ and $(ii)$.
\end{proof}

We present a basic but useful result about centralizers in $G^*$.\footnote{Proposition \ref{centralizza-cicloni}$(i)$ is well known, but we could not find an exact reference. So, for the sake of a self-contained exposition, we have preferred to prove it.}
 \begin{proposition}\label{centralizza-cicloni}
Let $\varphi\in \mathrm{Sym}(I)$ be a $2n$-cycle. Then, the following facts hold true:
\begin{itemize}
\item[$(i)$] $C_{\mathrm{Sym}(I)}(\langle\varphi\rangle)=\langle\varphi\rangle$;
\item[$(ii)$] if $\varphi\in G^*$, then $C_{G^*}(\langle\varphi\rangle)=\langle\varphi\rangle$;
\item[$(iii)$] if $\varphi\in G^*\setminus G$ and $n$ is odd, then $\varphi^n$ is the unique element of order $2$ in $C_{G^*}(\langle\varphi\rangle)\cap (G^*\setminus G)$;
\item[$(iv)$] if $\varphi\in G^*\setminus G$ and $n$ is even, then there exists no element of order $2$ in $C_{G^*}(\langle\varphi\rangle)\cap (G^*\setminus G)$.
\end{itemize}
\end{proposition}
\begin{proof} $(i)$ By the fact that $\langle\varphi\rangle$ is abelian, we immediately get that $C\coloneq C_{\mathrm{Sym}(I)}(\langle\varphi\rangle)\geq \langle\varphi\rangle$. Let us prove that $C$ is semiregular. Consider then $z\in I$ and $\psi \in \mathrm{Stab}_{C}(z)$ and prove that $\psi=id_I$. Given $u\in I$, since $\langle\varphi\rangle$ is transitive, we know that
there exists $k\in\ldbrack 2n\rdbrack$ such that $u=\varphi^k(z)$. Thus,
$$\psi(u)=\psi\varphi^k(z)=\varphi^k\psi(z)=\varphi^k(z)=u.$$
We conclude then that $\psi=id_I$, as desired.
$C$ is also transitive because it contains the transitive group $\langle\varphi\rangle$. Hence, $C$ is regular. By Proposition \ref{facili-semireg}$(iii)$, we deduce that $|C|=2n$ and hence $C=\langle\varphi\rangle$.

 $(ii)$ Let $\varphi\in G^*$. By $(i)$, we have that $C_{G^*}(\langle\varphi\rangle)= G^*\cap C_{\mathrm{Sym}(I)}(\langle\varphi\rangle)=G^*\cap \langle\varphi\rangle=\langle\varphi\rangle$.

$(iii)$ Let $\varphi\in G^*\setminus G$ and $n$ be odd. By $(ii)$ the only element of order $2$ in $C_{G^*}(\langle\varphi\rangle)$ is $\varphi^n$. Since $n$ is odd we also have that  $\varphi^n\in G^*\setminus G$.

$(iv)$ Let $\varphi\in G^*\setminus G$ and $n$ be even. By $(ii)$ the only element of order $2$ in $C_{G^*}(\langle\varphi\rangle)$ is $\varphi^n$. Since $n$ is even we have that $\varphi^n\in G$ and thus $\varphi^n\not\in G^*\setminus G$.
\end{proof}

Let us prove the main result of the section.

\begin{theorem}\label{main-semireg2} Let $n$ be odd and $S\leq G^*$ be semiregular. Then, there exists $\varphi\in C_{G^*}(S)\cap(G^*\setminus G)$ such that $|\varphi|=2$.
\end{theorem}

\begin{proof}  We have to show that there exists $\varphi\in G^*\setminus G$ satisfying the two following properties:
\begin{itemize}
\item[$(i)$]for every $s\in S$ and $z\in I$, $\varphi s (z)=s\varphi(z)$;
\item[$(ii)$]for every $z\in I$, $\varphi\varphi(z)=z$.
\end{itemize}
We will define $\varphi$ on $I$ by specifying its definition on the orbits of $S$ on $I$. To that purpose, we first observe that, given $z,z'\in I$ with $z'\in z^S$, there exists a unique $s\in S$ such that $z'=s(z)$. The existence of $s\in S$ for which $z'=s(z)$ follows directly from the definition of $S$-orbit of $z$. Furthermore, by Proposition \ref{facili-semireg}, if $s,s'\in S$ are such that $z'=s(z)=s'(z)$, then $s=s'$. This guarantees the uniqueness. As a consequence, given $z\in I$, we can uniquely represent every element of $z^S$ as $s(z)$ where $s\in S$. We will refer to such notable property as the {\it unique representation property} throughout the proof.
The proof is structured in two cases.

\vspace{1mm}

\noindent {\it Case 1. $S\le G$.}

\noindent By Proposition \ref{orbiteS1}, we know that every orbit of $S$ is included in $W$ or in $M$ and that  there are $r\coloneq\frac{n}{|S|}$ orbits of $S$ included in $W$ as well as $r$ orbits of $S$ included in $M$, for a total of $2r$ orbits. Let $O^w_1,\ldots,O^w_r$ be the distinct orbits of $S$ included in $W$ and $O^m_1,\ldots,O^m_r$ be the distinct orbits of $S$ included in $M$. For every $j\in \ldbrack r \rdbrack$, let us fix $x_j\in O^w_j$ and $y_j\in O^m_j$. Thus, for every $j\in \ldbrack r \rdbrack$, $O^w_j=x_j^S$ and $O^m_j=y_j^S$.

For every $j\in \ldbrack r \rdbrack$, we define $\varphi$ on $x_j^S\cup y_j^S$ by setting, for every $s\in S$,
$\varphi(s(x_j))\coloneq s(y_j)$ and $\varphi(s(y_j))\coloneq s(x_j)$.
Such a definition is consistent thanks to the unique representation property.
Note that, for every $z\in x_j^S$, $\varphi(z)\in y_j^S$ and, for every $z\in y_j^S$, $\varphi(z)\in x_j^S$. Indeed, if $z\in x_j^S$, then $z=s(x_j)$ for some $s\in S$ and $\varphi(z)=\varphi(s(x_j))=s(y_j)\in y_j^S$. Analogously, if $z\in y_j^S$, then $z=s(y_j)$ for some $s\in S$ and $\varphi(z)=\varphi(s(y_j))=s(x_j)\in x_j^S$. As a consequence, for every $z\in x_j^S\cup y_j^S$, we have that $\varphi(z)\in x_j^S\cup y_j^S$. Moreover, $\varphi:x_j^S\cup y_j^S\rightarrow x_j^S\cup y_j^S$ is a bijection. In fact, it is sufficient to prove that $\varphi:x_j^S\cup y_j^S\rightarrow x_j^S\cup y_j^S$ is injective. Consider then $z,z'\in x_j^S\cup y_j^S$ such that $\varphi(z)=\varphi(z')$ and prove that $z=z'$. By the property of $\varphi$ proved before, we must have $z,z'\in x_j^S$ or $z,z'\in y_j^S$. Assume $z,z'\in x_j^S$.
By definition of orbit, there exist $s,s'\in S$ such that $z=s(x_j)$ and $z'=s'(x_j)$. Thus, $\varphi(z)=\varphi(z')$  implies  $\varphi(s(x_j))=\varphi(s'(x_j))$ that in turn, by the definition of $\varphi$, implies $s(y_j)=s'(y_j)$. By the unique representation property, we deduce that $s=s'$ and therefore $z=s(x_j)=s'(x_j)=z'$.
The case where $z,z'\in y_j^S$ is analogous and thus omitted.

Let us prove that $\varphi\in G^*\setminus G$. It is clear that $\varphi\in \mathrm{Sym}(I)$. Thus, we only need to show that $\varphi(W)=M$. Consider then $x\in W$. We know that there exists  $j\in \ldbrack r \rdbrack$ such that $x\in x_j^S$. Then, $\varphi(x)\in y_j^S\subseteq M$. That shows $\varphi(W)\subseteq M$ and, since $|W|=|M|$, we deduce $\varphi(W)=M$.

Let us prove that $\varphi$ satisfies $(i)$. Consider $s\in S$ and $z\in I$, and let us show that $\varphi s (z)=s\varphi(z)$. If $z\in W$, there exists $j\in \ldbrack r \rdbrack$ such that $z\in x_j^S$. Then, $z=s'(x_j)$ for some $s'\in S$.  By definition of $\varphi$, we then have $\varphi s(z)=\varphi s s'(x_j)=s s'(y_j)$ and $s\varphi(z)=s\varphi (s'(x_j))=ss'(y_j)$; therefore $\varphi s (z)=s\varphi(z)$, as desired. If $z\in M$, there exists $j\in \ldbrack r \rdbrack$ such that $z\in y_j^S$. Then $z=s'(y_j)$ for some $s'\in S$.  By definition of $\varphi$, we then have $\varphi s(z)=\varphi s s'(y_j)=s s'(x_j)$ and $s\varphi(z)=s\varphi (s'(y_j))=ss(x_j)$; therefore $\varphi s (z)=s\varphi(z)$, as desired.

Let us finally prove that $\varphi$ satisfies $(ii)$. Consider $z\in I$. Assume first that $z\in W$. Then, there exists $j\in \ldbrack r \rdbrack$ such that $z\in x_j^S$. Let $s\in S$ be such that $z=s(x_j)$. Thus
\[
\varphi\varphi(z)=\varphi\varphi s(x_j)=\varphi s(y_j)=s (x_j)=z.
\]
Assume now that $z\in M$. Then, there exists $j\in \ldbrack r \rdbrack$ such that $z\in y_j^S$. Let $s\in S$ be such that $z=s(y_j)$. Thus
\[
\varphi\varphi(z)=\varphi\varphi s(y_j)=\varphi s(x_j)=s (y_j)=z.
\]
This completes the proof.

\vspace{1mm}

\noindent {\it Case 2. $S\not \le G$.}

\noindent Let $O_1,\ldots,O_r$ be the orbits of $S$, where $r\in\mathbb{N}$. By  Proposition \ref{orbiteS2}, we know that $|S|$ is even and that, for every $j\in \ldbrack r \rdbrack$, $|O_j\cap W|=|O_j\cap M|=\frac{|S|}{2}$. Since $2$ divides $|S|$, by Cauchy theorem, there exists $s^*\in S$ such that $|s^*|=2$. Let us prove that $s^*\in G^*\setminus G$. Indeed, by Proposition \ref{facili-semireg}, $\langle s^*\rangle$ is semiregular. Assume, by contradiction, that $\langle s^*\rangle\le G$. Then, since $n$ is odd, we can apply Proposition \ref{orbiteS1} to $\langle s^*\rangle$ and deduce that $|s^*|=|\langle s^*\rangle|=2$ divides $n$, a contradiction. For every $j\in \ldbrack r \rdbrack$, let us fix $x_j\in O_j\cap W$ and set $y_j\coloneq s^*(x_j)\in O_j\cap M$. Note that, for every $j\in \ldbrack r \rdbrack$, $x_j\coloneq s^*(y_j)$.

For every $j\in \ldbrack r \rdbrack$, we define $\varphi$ on $x_j^S$ by setting, for every $s\in S$, $\varphi(s(x_j))\coloneq s(y_j)$. Such a definition is consistent thanks to the unique representation property. Note that, for every $z\in x_j^S$, $\varphi(z)\in x_j^S$. Indeed, consider $z\in x_j^S$. Then $z=s(x_j)$ for some $s\in S$ and $\varphi(z)=\varphi(s(x_j))=s(y_j)=s s^*(x_j)\in x_j^S$. Moreover, $\varphi:x_j^S\rightarrow x_j^S$ is a bijection. In order to prove that fact it is sufficient to prove that $\varphi:x_j^S\rightarrow x_j^S$ is injective. Consider then $z,z'\in x_j^S$ such that $\varphi(z)=\varphi(z')$. We know that there are $s,s'\in S$ such that $z=s(x_j)$ and $z'=s'(x_j)$. Thus, $\varphi(z)=\varphi(z')$ is equivalent to $\varphi(s(x_j))=\varphi(s'(x_j))$ that in turn is equivalent to $s(y_j)=s'(y_j)$. By the unique representation property, we deduce that $s=s'$; therefore $z=s(x_j)=s'(x_j)=z'$. Thus, we conclude that $\varphi:x_j^S\rightarrow x_j^S$ is injective and then bijective.

We show that $\varphi\in G^*\setminus G$. It is immediate to prove that $\varphi\in \mathrm{Sym}(I)$.  Thus, it remains to show that $\varphi(W)=M$. Consider then $z\in W$ and prove that $\varphi(z)\in M$. We know that there exists $j\in \ldbrack r \rdbrack$ such that $z\in x_j^S$. Let $s\in S$ be such that $z=s(x_j)$. Since $z,x_j\in W$, we have that $s\in G$. Then, since $y_j\in M$, we have that $\varphi(z)=s(y_j)\in M$. That shows that $\varphi(W)\subseteq M$ and, since $|W|=|M|$, we deduce $\varphi(W)=M$.

We next prove that $\varphi$ satisfies $(i)$. Consider then $s\in S$ and $z\in I$, and show that $\varphi s (z)=s\varphi(z)$. Let $j\in \ldbrack r \rdbrack$ be such that $z\in x_j^S$.
Then $z=s'(x_j)$ for some $s'\in S$.  By definition of $\varphi$, we then have $\varphi s(z)=\varphi s s'(x_j)=s s'(y_j)$ and $s\varphi(z)=s\varphi (s'(x_j))=ss'(y_j)$; therefore $\varphi s (z)=s\varphi(z)$, as desired.  Thus, $(i)$ is proved.

We finally prove that $\varphi$ satisfies $(ii)$. Consider $z\in I$.
We know that there exists $j\in \ldbrack r \rdbrack$ such that $z\in x_j^S$. Let $s\in S$ be such that $z=s(x_j)$. Thus
$
\varphi\varphi(z)=\varphi\varphi s(x_j)=\varphi s(y_j)=\varphi s s^*(x_j)=s s^*(y_j)=s(x_j)=z.
$
\end{proof}

\subsection{Semiregularity of the $G^*$-stabilizers}\label{sr-stab}

In this section we analyze the $G^*$-stabilizers of preference profiles  with the main purpose to show that they are always semiregular subgroups of $G^*$ (Theorem \ref{stabilizzatori-semireg}). The next proposition is the main auxiliary result of the section: it characterizes the types of permutations belonging to the $G^*$-stabilizers of a preference profile. For its proof we rely on the content of Section \ref{GT1}.

\begin{proposition}\label{all-types}
Let $p\in \mathcal{P}$ and $\varphi\in \mathrm{Stab}_{G^*}(p)$.  Then, the following facts hold true:
\begin{itemize}
\item[$(i)$] if $\varphi\in G$, then $T_{\varphi}=[b^{(r)}]$, where $b,r\in\mathbb{N}$, $b=|\varphi|$ and $br=2n$;
\item[$(ii)$] if $\varphi\in G^*\setminus G$, then $T_{\varphi}=[b^{(r)}]$, where $b,r\in\mathbb{N}$, $b=|\varphi|$ is even and  $br=2n$.
\end{itemize}
\end{proposition}

\begin{proof}
$(i)$ Let $\varphi\in G$. Then $\varphi$ takes the set $W$ into $W$ and the set $M$ into $M$. As a consequence, each orbit of
$\varphi$ on $I$ is included in one of the sets $W$ and $M$. 
Let $x_1,\ldots,x_k,y_1,\ldots,y_h$ be representatives of the orbits of $\varphi$ on $I$, where $k,h\ge 1$, $x_1,\ldots,x_k\in W$ and $y_1,\ldots,y_h\in M$.
Let us set, for every $j\in \ldbrack k \rdbrack$, $b_j\coloneq|x_j^{\langle\varphi\rangle}|$ and, for every $i\in \ldbrack h \rdbrack$,
$c_i\coloneq|y_i^{\langle\varphi\rangle}|$.

Pick $j\in \ldbrack k \rdbrack$. Since $\mathrm{Stab}_{G^*}(p)\cap G$ is a group, we have that $\varphi^{b_j}\in \mathrm{Stab}_{G^*}(p)\cap G$, and then, by \eqref{action-w2}, we deduce
\begin{equation}\label{orbite2}
\varphi^{b_j} p(x_j)=p(\varphi^{b_j}(x_j)).
\end{equation}
Using the fact that $x_j$ is fixed by $\varphi^{b_j}$, \eqref{orbite2} implies $\varphi^{b_j} p(x_j)=p(x_j)$. Since $p(x_j)\in \mathbf{L}(M)$, by Lemma \ref{phiRR}, we get that, for every $y\in M$, $\varphi^{b_j}(y)=y$. 
Consider now $\nu\in \mathrm{Sym}(M)$ defined, for every $y\in M$, by $\nu(y)=\varphi(y)$. For every $y\in M$, we have that $\nu^{b_j}(y)=\varphi^{b_j}(y)=y$. Thus, $\nu^{b_j}=id_M$  which implies $|\nu|$ divides $b_j$. Moreover, $y_1,\ldots,y_h\in M$ are representatives of the orbits of $\nu$ on $M$ and, for every $i\in\ldbrack h\rdbrack$, we have $y_i^{\langle\nu\rangle}=y_i^{\langle\varphi\rangle}$. As a consequence, for every $i\in \ldbrack h \rdbrack$, $c_i=|y_i^{\langle\nu\rangle}|$ divides $|\nu|$. Thus, for every $i\in \ldbrack h \rdbrack$, $c_i$ divides $b_j$.

Pick next $i\in \ldbrack h \rdbrack$. Since $\mathrm{Stab}_{G^*}(p)\cap G$ is a group, we have that $\varphi^{c_i}\in \mathrm{Stab}_{G^*}(p)\cap G$, and then, by \eqref{action-w2}, we deduce
\begin{equation}\label{orbite3}
\varphi^{c_i} p(y_i)=p(\varphi^{c_i}(y_i)).
\end{equation}
Using the fact that $y_i$ is fixed by $\varphi^{c_i}$, \eqref{orbite3} implies $\varphi^{c_i} p(y_i)=p(y_i)$. Since $p(y_i)\in \mathbf{L}(W)$, by Lemma \ref{phiRR}, we get that, for every $x\in W$, $\varphi^{c_i}(x)=x$. 
Consider now $\sigma\in \mathrm{Sym}(W)$ defined, for every $x\in W$, by $\sigma(x)=\varphi(x)$. For every $x\in W$, we have that $\sigma^{c_i}(x)=\varphi^{c_i}(x)=x$.
Thus, $\sigma^{c_i}=id_W$ and that implies that $|\sigma|$ divides $c_i$. Moreover, $x_1,\ldots,x_k\in W$ are representatives of the orbits of $\sigma$ on $W$ and, for every $j\in\ldbrack k\rdbrack$, we have $x_j^{\langle\sigma\rangle}=x_j^{\langle\varphi\rangle}$. As a consequence, for every $j\in \ldbrack k \rdbrack$, $b_j=|x_j^{\langle\sigma\rangle}|$ divides $|\sigma|$. Thus, for every $j\in \ldbrack k \rdbrack$, $b_j$ divides $c_i$.

We finally deduce that, for every $j\in\ldbrack k \rdbrack$ and $i\in\ldbrack h \rdbrack$, $b_j$ divides $c_i$ and $c_i$ divides $b_j$, and therefore $b_j=c_i$.  Thus, the parts of $\varphi$ are equal to each other and all equal to $|\varphi|$. Clearly, $|\varphi|r=2n$, where $r=k+h$.

$(ii)$ Assume that $\varphi\in G^*\setminus G$. Let $T_{\varphi}=[b_1,\dots, b_r]$, where $r\in\mathbb{N}$ and $b_j\in\mathbb{N}$ for all $j\in \ldbrack r \rdbrack$.
Observe that a permutation in $G^*\setminus G$ takes $W$ into $M$ and $M$ into $W$. Thus, every orbit of $\varphi$ on $I$ contains at least an element in $W$ and an element in $M$. Moreover, the number of women and men in an orbit is the same. As a consequence, we have that, for every $j\in \ldbrack r \rdbrack$, $b_j\geq 2$ and $b_j$ is even.
Let $z_1,\ldots,z_r\in I$ be representatives of the orbits of $\varphi$ on $I$. Pick $j\in \ldbrack r \rdbrack$. Since $\mathrm{Stab}_{G^*}(p)$ is a group, we have that $\varphi^{b_j}\in \mathrm{Stab}_{G^*}(p)$, and then, by \eqref{action-w2}, we deduce that,  for every $z\in I$,
\begin{equation}\label{orbite}
\varphi^{b_j} p(z)=p(\varphi^{b_j}(z)).
\end{equation}
Consider $x\in z_j^{\langle\varphi\rangle}\cap W$. Then, $x$ is fixed by $\varphi^{b_j}$, and thus \eqref{orbite} implies $\varphi^{b_j} p(x)=p(x)$. Since $p(x)\in \mathbf{L}(M)$, by Lemma \ref{phiRR}, we deduce that, for every $y\in M$, $\varphi^{b_j}(y)=y$. Consider now $y\in z_j^{\langle\varphi\rangle}\cap M$.
Then $y$ is fixed by $\varphi^{b_j}$, and thus \eqref{orbite} implies $\varphi^{b_j} p(y)=p(y)$. Since $p(y)\in \mathbf{L}(W)$, by Lemma \ref{phiRR}, we deduce that, for every $x\in W$, $\varphi^{b_j}(x)=x$. It follows that $\varphi^{b_j}=id_I$. As a consequence, $|\varphi|$ divides $b_j$.
Since we know that $b_j$ divides $|\varphi|$, we conclude that $|\varphi|= b_j$. Since that holds true for any $j\in\ldbrack r\rdbrack$, we obtain that the parts in $T_{\varphi}$ are equal to each other and all equal to $|\varphi|$. Clearly, $|\varphi|r=2n$.
\end{proof}

\begin{theorem}\label{stabilizzatori-semireg}
Let $p\in \mathcal{P}$. Then, $\mathrm{Stab}_{G^*}(p)$ is a semiregular subgroup of $G^*$.
\end{theorem}

\begin{proof} Let $S\coloneq\mathrm{Stab}_{G^*}(p)\leq \mathrm{Sym}(I)$. We need to show that, for every $z\in I$,
$\mathrm{Stab}_S(z)=\{id_I\}$. Assume, by contradiction, that there exist $z\in I$ and $\varphi\in S\setminus\{id_I\}$ such that
$\varphi(z)=z$.
By the fact that $\varphi \in S\setminus\{id_I\}$ and by Proposition \ref{all-types}, we deduce that $T_{\varphi}=[b^{(r)}]$, with $b=|\varphi|\geq 2$. But such a permutation does not fix any element of $I$, a contradiction.
\end{proof}

We now describe, up to group isomorphisms,  the whole family of stabilizers when $n=3$. That description helps the reader to understand the many possibilities that can arise and the fact that the description of the stabilizers is not, in general, an immediate task. 

Let $n=3$ and consider $p\in \mathcal{P}$. By Theorem \ref{stabilizzatori-semireg} and Proposition \ref{facili-semireg}, we know that the size of $\mathrm{Stab}_{G^*}(p)$ divides 6 and hence it belongs to the set $\{1,2,3,6\}$.
As a consequence, up to isomorphisms, we have $\mathrm{Stab}_{G^*}(p)\in\{\{ id_I\},C_2, C_3, C_6, \mathrm{Sym}(\ldbrack 3\rdbrack)\}$.
Interestingly, all those cases can arise. Indeed, consider the following preference profiles
\[
p_1=\begin{array}{|ccc||ccc|}
\hline
1&2&3&4&5&6\\
\hline
\hline
4&4&4&1&1&1\\
5&5&5&2&2&3\\
6&6&6&3&3&2\\
\hline
\end{array}\,,
\quad
p_2=\begin{array}{|ccc||ccc|}
\hline
1&2&3&4&5&6\\
\hline
\hline
4&4&4&1&1&1\\
5&5&5&2&2&2\\
6&6&6&3&3&3\\
\hline
\end{array}\,,
\quad
p_3=\begin{array}{|ccc||ccc|}
\hline
1&2&3&4&5&6\\
\hline
\hline
4&5&6&2&3&1\\
5&6&4&1&2&3\\
6&4&5&3&1&2\\
\hline
\end{array}\,,
\]
\[
p_4=\begin{array}{|ccc||ccc|}
\hline
1&2&3&4&5&6\\
\hline
\hline
4&5&6&1&2&3\\
5&6&4&2&3&1\\
6&4&5&3&1&2\\
\hline
\end{array}\,,
\quad
p_5=\begin{array}{|ccc||ccc|}
\hline
1&2&3&4&5&6\\
\hline
\hline
4&5&6&1&2&3\\
5&6&4&3&1&2\\
6&4&5&2&3&1\\
\hline
\end{array}\,.
\]
A computation shows that
\[
\begin{array}{l}
\mathrm{Stab}_{G^*}(p_1)=\{id_I\},\\
\vspace{-2mm}\\
\mathrm{Stab}_{G^*}(p_2)= \langle(14)(25)(36)\rangle\cong C_2,\\
\vspace{-2mm}\\
\mathrm{Stab}_{G^*}(p_3)=\langle (123)(456)\rangle\cong C_3,\\
\vspace{-2mm}\\
\mathrm{Stab}_{G^*}(p_4)= \langle (153426)\rangle\cong C_{6},\\
\vspace{-2mm}\\
\mathrm{Stab}_{G^*}(p_5)=\{id_I, (123)(456),(132)(465), (14)(26)(35), (15)(24)(36),(16)(25)(34)\}\cong \mathrm{Sym}(\ldbrack 3\rdbrack).
\end{array}
\]

In light of Theorem \ref{F-ref-consistent}, preference profiles  with large stabilizers are crucial for establishing impossibility results. Indeed, if $p_1, p_2\in \mathcal{P}$ are such that $\mathrm{Stab}_{G^*}(p_1)\supseteq \mathrm{Stab}_{G^*}(p_2),$
by \eqref{nuova-veste}, we get that $C^U(p_1)\subseteq C^U(p_2)$, for all $U\leq G^*$. As a consequence, $C^U(p_1)\cap F(p_1)\subseteq C^U(p_2)\cap F(p_2)$ and so $p_1$ is a better candidate than $p_2$ to get $C^U(p_1)\cap F(p_1)=\varnothing$.

By Theorem \ref{stabilizzatori-semireg} and Proposition \ref{facili-semireg}, we know that the size of a stabilizer cannot be larger than $2n$. Propositions \ref{ciclone} and \ref{ciclone-odd} shows that such an upper bound can be reached by a stabilizer that is cyclic. Moreover, in the odd case, an additional property can also be satisfied.

\begin{proposition}\label{ciclone}
There exists $p\in \mathcal{P}$ such that $\mathrm{Stab}_{G^*}(p)$ is generated by a $2n$-cycle.
\end{proposition}

\begin{proof}
Consider the $2n$-cycle $\varphi\coloneq(1\quad n+1\quad 2\quad n+2 \ \dots \ n\quad 2n)\in G^*$.
We are going to exhibit $p\in \mathcal{P}$ such that $\mathrm{Stab}_{G^*}(p)=\langle \varphi\rangle$. Note that $\varphi$ is a $2n$-cycle and $\varphi\in G^*\setminus G$. Thus $\varphi^2\in G$ and clearly $\varphi^2=\sigma\nu$, where $\sigma\coloneq(1\ 2\ \dots\  n)\in G_W$ and $\nu\coloneq(n+1\ n+2\ \dots\  2n)\in G_M$. Note that $|\sigma|=|\nu|=n$; for every $x\in W$, $\sigma(x)=\varphi^2(x)$; for every $y\in M$, $\nu(y)=\varphi^2(y)$. 

Define, for every $x\in W$,
\[
p(x)\coloneq[\nu^{x-1}(n+1),\dots, \nu^{x-1}(2n)]\in \mathbf{L}(M)
\]
and, for every $y\in M$,
\[
p(y)\coloneq[\sigma^{y}(1),\dots, \sigma^{y}(n)]\in \mathbf{L}(W).
\]
In order to prove that $\varphi\in \mathrm{Stab}_{G^*}(p)$, we show that, for every $z\in I$, we have
\begin{equation}\label{stab-phi}
\varphi p(z)=p(\varphi(z)).
\end{equation}
Let us consider then $z\in I$. Assume first that $z\in W$. Then
\[
\varphi p(z)=\varphi[\nu^{z-1}(n+1),\dots, \nu^{z-1}(2n)]=\varphi [\varphi^{2z-2}(n+1),\dots, \varphi^{2z-2}(2n)]
\]
\[
=[\varphi^{2z-1}(n+1),\dots, \varphi^{2z-1}(2n)].
\]
On the other hand, we also have
\[
p(\varphi(z))=p(n+z)=[\sigma^{n+z}(1),\dots, \sigma^{n+z}(n)]=[\sigma^{z}(1),\dots, \sigma^{z}(n)]=[\varphi^{2z}(1),\dots, \varphi^{2z}(n)]
\]
\[
=[\varphi^{2z-1}\varphi(1),\dots, \varphi^{2z-1}\varphi(n)]=[\varphi^{2z-1}(n+1),\dots, \varphi^{2z-1}(2n)],
\]
and hence \eqref{stab-phi} follows.

Assume next that $z\in M$. Then,  we have
\[
\varphi p(z)=\varphi[\sigma^{z}(1),\dots, \sigma^{z}(n)]=\varphi[\varphi^{2z}(1),\dots, \varphi^{2z}(n)]
=[\varphi^{2z}(n+1),\dots, \varphi^{2z}(2n)].
\]
In particular, we have $\varphi p(2n)=[n+1,\dots, 2n]$.
On the other hand, for $n+1\leq z<2n$, we also have
\[
p(\varphi(z))=p(z-n+1)=[\nu^{z-n}(n+1),\dots, \nu^{z-n}(2n)]
\]
\[
=[\nu^{z}(n+1),\dots, \nu^{z}(2n)]=[\varphi^{2z}(n+1),\dots, \varphi^{2z}(2n)].
\]
Finally, observe that $p(\varphi(2n))=p(1)=[n+1,\dots, 2n]=\varphi p(2n).$
Thus, \eqref{stab-phi} is proved.

Now from $\mathrm{Stab}_{G^*}(p)\geq \langle \varphi \rangle$, we deduce that $2n$ divides $\mathrm{Stab}_{G^*}(p)$. By Theorem \ref{stabilizzatori-semireg} and  Proposition \ref{facili-semireg}, we know that $|\mathrm{Stab}_{G^*}(p)|$ divides $2n$ and so
$\mathrm{Stab}_{G^*}(p)=\langle \varphi \rangle$.
\end{proof}

\begin{proposition}\label{ciclone-odd} Let $n$ be odd. Then there exists $p\in \mathcal{P}$ such that $\mathrm{Stab}_{G^*}(p)$ is generated by a $2n$-cycle $\varphi\in G^*$ with the property that, for every $z\in I$, $\mathrm{Rank}_{p(z)}(\varphi^n(z))=n$.
\end{proposition}

\begin{proof}
Let $n\geq 3$ be odd. Consider the $2n$-cycle $\varphi\coloneq(1\quad n+1\quad 2\quad n+2 \ \dots \ n\quad 2n)\in G^*$.
Note that, for every $x\in W$, we have $\varphi(x)=x+n$; for every $y\in M$, $\varphi(y)$ is equal to the unique element in $[y-n+1]_n\cap \ldbrack n \rdbrack$.

We claim that, for every $x\in W$ odd, $\varphi^x(1)=n+\frac{x+1}{2}\in M$.  We prove  that by finite induction on $x$.
Surely $\varphi(1)=n+\frac{1+1}{2}=n+1$. Assume that the thesis holds for an odd $x\leq n-2$ and show it for the next odd, that is, for $x+2$.
Observe also that every $x\in W$ satisfies $\frac{x+1}{2}< \frac{x+3}{2}\leq n,$ because $n\geq 3$. Those arithmetic considerations and the inductive assumption imply then that
\[
\varphi^{x+2}(1)= \varphi^2\varphi^{x}(1)=\varphi^2\left(n+\frac{x+1}{2}\right)=\varphi\left(\frac{x+1}{2}+1\right)=\varphi\left(\frac{x+3}{2}\right)=n+\frac{x+3}{2},
\]
which concludes the proof of the claim.
As a consequence, using the fact that $n$ is odd, we get $\varphi^n(1)=\frac{3n+1}{2}.$

We claim next that, for every $x\in W$,  $\varphi^{2(x-1)}(1)=x$. We prove  that by finite induction on $x$, for $x\in W$.
If $x=1$, then $\varphi^{2(1-1)}(1)=\varphi^0(1)=id_I(1)=1$. Assume that the thesis holds for $x\leq n-1$. Then, we have
\[
\varphi^{2x}(1)= \varphi^2\varphi^{2(x-1)}(1)=\varphi^2(x)=\varphi(n+x)=x+1,
\]
and hence the thesis holds for  $x+1$.

Define now $p\in \mathcal{P}$ as follows. Let $p(1)=\left[y_1,\dots, y_{n-1}, \frac{3n+1}{2}\right]$, where
$y_1,\dots, y_{n-1}\in M$ are such that
$M=\left\{y_1,\dots, y_{n-1}, \frac{3n+1}{2}\right\}$. Define then, for every $x\in W$, $p(x)=\varphi^{2(x-1)}p(1)$ and, for every $y\in M$, $p(y)=\varphi^{2y-2n-1}p(1)$. We show that $\varphi\in \mathrm{Stab}_{G^*}(p)$, that is, for every $z\in I$, we have
\begin{equation}\label{stab-phi2}
\varphi p(z)=p(\varphi(z)).
\end{equation}
Let us consider $z\in I$. Assume first that $z\in W$. Then we have
\[
\varphi p(z)=\varphi \varphi^{2(z-1)}p(1)=\varphi^{2z-1}p(1).
\]
On the other hand, we also have
\[
p(\varphi(z))=p(n+z)=\varphi^{2(n+z)-2n-1}p(1)=\varphi^{2z-1}p(1),
\]
and so \eqref{stab-phi2} is satisfied.

Assume next that $z\in M$. Then, we have
\[
\varphi p(z)=\varphi \varphi^{2z-2n-1}p(1)=\varphi^{2z-2n}p(1).
\]
In particular, $\varphi p(2n)=p(1)$. On the other hand, we also have that
$p(\varphi(z))$ equals the preferences, with respect to $p$, of the only individual in the set $[z-n+1]_n\cap \ldbrack n \rdbrack.$
Indeed, if $z=2n$ that gives $p(\varphi(2n))=p(1)$; if $n+1\le z\leq 2n-1$, that gives
\[
p(\varphi(z))=p(z-n+1)=\varphi^{2(z-n+1-1)}p(1)=\varphi^{2z-2n}p(1).
\]
Hence \eqref{stab-phi2} is proved.

Let now $z\in I$. We have to show that $\mathrm{Rank}_{p(z)}(\varphi^n(z))=n$.
Assume first that $z\in W$. By definition of $p$ and by the properties observed for $\varphi$, we have that
\[
\mathrm{Rank}_{p(z)}\left(\varphi^{2(z-1)}\left(\mbox{$\frac{3n+1}{2}$}\right)\right)=n,
\]
and
\[
\varphi^{2(z-1)}\left(\mbox{$\frac{3n+1}{2}$}\right)=\varphi^{2(z-1)}\varphi^{n}(1)=\varphi^{n}\varphi^{2(z-1)}(1)=\varphi^n(z).
\]
Therefore, $\mathrm{Rank}_{p(z)}(\varphi^n(z))=n$.
Assume next that $z\in M$. Observe that $\varphi^n \in G^*\setminus G$ and $|\varphi^n|=2$. Consider then $x\in W$ such that $z=\varphi^n(x)$ and $x=\varphi^n(z)$. Since $\varphi^n\in \mathrm{Stab}_{G^*}(p)$, we know that $p(z)=p(\varphi^n(x))=\varphi^n p(x)$.
Since $x\in W$, we know that $\mathrm{Rank}_{p(x)}(\varphi^n(x))=n$. Thus, we conclude that
\[
\mathrm{Rank}_{p(z)}(\varphi^n(z))=\mathrm{Rank}_{p(z)}(x)=\mathrm{Rank}_{\varphi^n p(x)}(\varphi^n\varphi^n(x))=\mathrm{Rank}_{p(x)}(\varphi^n(x))=n.
\]

Now, from $\mathrm{Stab}_{G^*}(p)\geq \langle \varphi \rangle$, we deduce that $2n$ divides $|\mathrm{Stab}_{G^*}(p)|$. By Theorem \ref{stabilizzatori-semireg} and Proposition \ref{facili-semireg}, we know that $|\mathrm{Stab}_{G^*}(p)|$ divides $2n$. As a consequence,
$\mathrm{Stab}_{G^*}(p)=\langle \varphi \rangle$.
 \end{proof}

\section{Proofs of the results of Section \ref{PossImpossThs} and Section \ref{outside}}\label{main-proofs}

In this section, we finally prove the results stated in Section \ref{PossImpossThs} and Section \ref{outside}. The proofs proposed strongly rely on the results proved in Appendix \ref{semi-reg}.

First of all, note that, using \eqref{nuova-veste}, we have that, for every preference profile $p\in \mathcal{P}$,
\begin{equation}\label{imp}
C^{G^*}(p)=\{\mu\in C_{G^*}(\mathrm{Stab}_{G^*}(p))\cap(G^*\setminus G): |\mu|=2\}.
\end{equation}

\begin{proof}[Proof of Theorem \ref{n2}] By Proposition \ref{ciclone}, there exists $p\in \mathcal{P}$ such that $\mathrm{Stab}_{G^*}(p)=\langle\varphi\rangle$, where $\varphi$ is a $2n$-cycle. By Corollary \ref{facile3}, we complete the proof by showing that
$C^{G^*}(p)=\varnothing$. By Proposition \ref{centralizza-cicloni}$(iv)$, recalling \eqref{imp}, we immediately have that $C^{G^*}(p)=\varnothing$.
\end{proof}

\begin{proof}[Proof of Theorem \ref{main-mathc}] By Corollary \ref{facile3}, it is sufficient to show that, for every $p\in \mathcal{P}$,
$C^{G^*}(p)\neq\varnothing$. Let $p\in \mathcal{P}$ and let $S\coloneq \mathrm{Stab}_{G^*}(p)$. By Theorem \ref{stabilizzatori-semireg}, $S$ is semiregular. Thus, recalling \eqref{imp} and applying Theorem \ref{main-semireg2}, we immediately get $C^{G^*}(p)\neq\varnothing$.
\end{proof}

\begin{proof}[Proof of Theorem \ref{no-sym-mo}] By Theorem \ref{F-ref-consistent}, it is sufficient to show that there exists $p\in \mathcal{P}$ such that $C^{G^*}(p)\cap MO(p)=\varnothing$. Since $n$ is odd, by Proposition \ref{ciclone-odd}, we know that there exists $p\in \mathcal{P}$ such that
$\mathrm{Stab}_{G^*}(p)=\langle\varphi\rangle$, where $\varphi\in G^*$ is a $2n$-cycle such that, for every $z\in I$, $\mathrm{Rank}_{p(z)}(\varphi^n(z))=n$. By Proposition \ref{centralizza-cicloni}$(iii)$, we have that the unique element of order $2$ in $C_{G^*}(\langle\varphi\rangle)\cap (G^*\setminus G)$
 is $\varphi^n$.
However, $\varphi^n\notin  MO(p)$ and thus, recalling \eqref{imp}, we get $C^{G^*}(p)\cap MO(p)=\varnothing$.
\end{proof}

In order to prove Theorem \ref{existence3}, we first need two preliminary results.

\begin{proposition}\label{prop44}
Let $\varphi\in\mathcal{M}$. Then $\{\varphi^\nu\in \mathrm{Sym}(I): \nu\in G_M\}=\mathcal{M}$.
\end{proposition}

\begin{proof}
By Proposition \ref{car-match}, we know that $\{\varphi^\nu\in \mathrm{Sym}(I): \nu\in G_M\}\subseteq\mathcal{M}$.
Consider now $\mu\in\mathcal{M}$. Define $\nu\in G_M$ by $\nu(x)=x$, for all $x\in W$,  and $\nu(y)=\mu\varphi^{-1}(y)$, for all $y\in M$. Note that $\varphi^{-1}=\varphi$, $\mu^{-1}=\mu$. Moreover, $\nu^{-1}(x)=x$, for all $x\in W$,  and $\nu^{-1}(y)=\varphi\mu^{-1}(y)$, for all $y\in M$.
Then, we have that $\varphi^\nu=\mu$. Indeed, let $z\in I$. If $z\in W$, then
\[
\varphi^\nu(z)=\nu\varphi\nu^{-1}(z)=\nu\varphi(z)=\mu\varphi^{-1}\varphi(z)=\mu(z).
\]
If $z\in M$, then
\[
\varphi^\nu(z)=\nu\varphi\nu^{-1}(z)=\nu\varphi\varphi\mu^{-1}(z)=\nu\mu^{-1}(z)=\nu\mu(z)=\mu(z).
\]
Thus $\mu\in\{\varphi^\nu\in \mathrm{Sym}(I): \nu\in G_M\}$.
\end{proof}

\begin{proposition}\label{matchcomm}
Let $\varphi,\mu\in\mathcal{M}$. Then $\varphi\mu=\mu\varphi$ if and only if there exists $\nu\in G_M$ with $|\nu|\le 2$ such that $\mu=\varphi^\nu$.
\end{proposition}

\begin{proof}Assume that there exists $\nu\in G_M$ with $|\nu|\le 2$ such that $\mu=\varphi^\nu$. Note that $\varphi=\varphi^{-1}$ and $\nu=\nu^{-1}$. We have to prove that, for every $z\in I$, $\varphi\mu(z)=\mu\varphi(z)$.
Let $z\in I$. If $z\in W$, then
\[
\varphi\mu(z)=\varphi\varphi^\nu(z)=\varphi\nu\varphi\nu^{-1}(z)=\varphi\nu\varphi(z)=\varphi\nu^{-1}\varphi(z)=\nu\varphi\nu^{-1}\varphi(z)=\mu\varphi(z).
\]
If $z\in M$, then
\[
\varphi\mu(z)=\varphi\varphi^\nu(z)=\varphi\nu\varphi\nu^{-1}(z)=\varphi\varphi\nu^{-1}(z)=\nu^{-1}(z)=\nu(z)=\nu\varphi\varphi(z)=\nu\varphi\nu^{-1}\varphi(z)=\mu\varphi(z).
\]

Conversely assume that $\varphi\mu=\mu\varphi$. By Proposition \ref{prop44}, we know that there exists $\nu\in G_M$ such that $\mu=\varphi^\nu$. Suppose, by contradiction, that $|\nu|\ge 3$. Then the cycle decomposition of $\nu$ contains a $k$-cycle $\gamma$, with $k\ge 3$, moving only men.
Let $S\subseteq M$ be the support of $\gamma$ and let $y_0\in S$. Then $|S|=k$, $S=y_0^{\langle \gamma\rangle }$, and $\gamma(y_0)\neq y_0$. We have that
\[
\varphi\mu(y_0)=\varphi\nu\varphi\nu^{-1}(y_0)=\varphi\varphi\nu^{-1}(y_0)=\nu^{-1}(y_0)=\gamma^{-1}(y_0),
\]
and
\[
\mu\varphi(y_0)=\nu\varphi\nu^{-1}\varphi(y_0)=\nu\varphi\varphi(y_0)=\nu(y_0)=\gamma(y_0).
\]
Since $\varphi\mu=\mu\varphi$, we must have $\gamma^{-1}(y_0)=\gamma(y_0)$. Thus, we deduce
$\gamma^2(y_0)=y_0$ and so $\{y_0,\gamma(y_0)\}=y_0^{\langle \gamma\rangle }$. It follows that $S=\{y_0,\gamma(y_0)\}$ and therefore $k=2$, a contradiction.
\end{proof}

\begin{proof}[Proof of Theorem \ref{existence3}]
Since $WPO$ is a $G^*$-symmetric matching mechanism, we have that $WPO$ is a $\langle\varphi\rangle$-symmetric matching mechanism. By Theorem \ref{F-ref-consistent}, it is sufficient to show that, for every $p\in\mathcal{P}$, $C^{\langle\varphi\rangle}(p)\cap WPO(p)\neq\varnothing$.
Consider then $p\in\mathcal{P}$. Observe that $\mathrm{Stab}_{\langle\varphi\rangle}(p)\in\{\{id_I\},\langle\varphi\rangle\}$. If
$\mathrm{Stab}_{\langle\varphi\rangle}(p)=\{id_I\}$, then $C^{\langle\varphi\rangle}(p)=\mathcal{M}$; therefore
$C^{\langle\varphi\rangle}(p)\cap WPO(p)=WPO(p)\neq\varnothing$.
Assume then that $\mathrm{Stab}_{\langle\varphi\rangle}(p)=\langle\varphi\rangle$.
By Proposition \ref{matchcomm}, we know that
\[
C^{\langle\varphi\rangle}(p)=\{\mu\in\mathcal{M}: \exists \nu\in G_M\mbox{ with }|\nu|\le 2\mbox{ such that } \mu=\varphi^\nu\}.
\]
In particular, $\varphi\in C^{\langle\varphi\rangle}(p)$.
Let $y_1\in M$ be the first ranked man in $p(1)$. If $\varphi(1)=y_1$, then $\varphi\in C^{\langle\varphi\rangle}(p)\cap WPO(p)\neq\varnothing$.  If $\varphi(1)=y_0\neq y_1$, consider $\nu=(y_0\,y_1)\in G_M$. Note that $|\nu|=2$ and that
$\varphi^\nu(1)=\nu\varphi(1)=\nu(y_0)=y_1$. Thus,
$\varphi^\nu\in C^{\langle\varphi\rangle}(p)\cap WPO(p)\neq\varnothing$.
\end{proof}

\begin{proof}[Proof of Theorem \ref{existence2}]
We have that $\mathcal{M}=\{\mu_1,\mu_2\}$, where $\mu_1\coloneq(13)(24)$ and $\mu_2\coloneq(14)(23)$. Note that $\mu_1\mu_2=\mu_2\mu_1$. Thus, for every $p\in\mathcal{P}$, we have that $C^{\langle\varphi\rangle}(p)=\mathcal{M}$.
As a consequence, for every $p\in\mathcal{P}$, $C^{\langle\varphi\rangle}(p)\cap ST(p)=ST(p)\neq\varnothing$. Then, by Theorem \ref{F-ref-consistent}, we conclude that there exists a resolute, stable and $\langle \varphi\rangle$-symmetric matching mechanism.
\end{proof}

\begin{proof}[Proof of Theorem \ref{endriss}] Define, for every $x\in W=\ldbrack n \rdbrack$,  $y_x\coloneq \varphi(x)$. Then,
$M=\{y_x: x\in \ldbrack n \rdbrack\}$. Let $p\in\mathcal{P}$ be defined by
\[
p(1)=[y_2,y_3,y_1,(y_4),\ldots,(y_n)],\quad p(y_1)=[2,3,1,(4),\ldots,(n)],
\]
\[
p(2)=[y_3,y_1,y_2,(y_4),\ldots,(y_n)],\quad p(y_2)=[3,1,2,(4),\ldots,(n)],
\]
\[
p(3)=[y_1,y_2,y_3,(y_4),\ldots,(y_n)],\quad p(y_3)=[1,2,3,(4),\ldots,(n)],
\]
and, for every $x\in \ldbrack n \rdbrack$ with $x\ge 4$ if any,
\[
p(x)=\sigma^{x-1}[y_1,y_2,y_3,(y_4),\ldots,(y_n)],\quad p(y_x)=\sigma^{x-1}[1,2,3,(4),\ldots,(n)],
\]
where $\sigma=(1\,2\dots n)(y_1\,y_2\dots y_n)$. It is immediate to check that $p^\varphi=p$, and thus $\varphi\in \mathrm{Stab}_{G^*}(p)$. Using Theorem \ref{F-ref-consistent}, we complete the proof by showing that $C^{\langle\varphi\rangle}(p)\cap ST(p)=\varnothing$.
Let $\gamma_1,\ldots,\gamma_6\in \mathrm{Sym}(I)$ be defined by
\[
\gamma_1=(1\,y_1)(2\,y_2)(3\,y_3), \quad \gamma_2=(1\,y_1)(2\,y_3)(3\,y_2) \quad \gamma_3=(1\,y_2)(2\,y_1)(3\,y_3),
\]
\[
\gamma_4=(1\,y_3)(2\,y_2)(3\,y_1), \quad\gamma_5=(1\,y_2)(2\,y_3)(3\,y_1), \quad\gamma_6=(1\,y_3)(2\,y_1)(3\,y_2),
\]
and let $\gamma^*=(4\,y_4)\cdots(n\,y_n)\in \mathrm{Sym}(I)$ if $n\ge 4$, while $\gamma^*=id_I$ if  $n= 3$. For every $i\in \ldbrack 6 \rdbrack$, set $\mu_i=\gamma_i \gamma^*\in\mathcal{M}$. Note  that $\varphi=\mu_1$ and that
\[
\{\mu\in\mathcal{M}: \forall x \in \ldbrack n \rdbrack\mbox{ with }x\ge 4, \mu(x)=y_x \}=\{\mu_1,\ldots,\mu_6\}.
\]
It is simple to verify that $ST(p)\subseteq \{\mu_1,\ldots,\mu_6\}$. Indeed, consider $\mu\in\mathcal{M}\setminus \{\mu_1,\ldots,\mu_6\}$. Then, there exists $x\in \ldbrack n \rdbrack$ with $x\ge 4$ such that $\mu(x)\neq y_x$. Thus, $\mu$ is blocked by $(x,y_x)$  and so $\mu\not\in ST(p)$.
Observe now that
$\mu_1$ is blocked by $(1,y_2)$;
$\mu_2$ is blocked by $(1,y_3)$;
$\mu_3$ is blocked by $(3,y_2)$;
$\mu_4$ is blocked by $(2,y_1)$.
As a consequence, $ST(p)\subseteq \{\mu_5,\mu_6\}$. However, $C^{\langle\varphi\rangle}(p )\cap\{\mu_5,\mu_6\}=\varnothing$ because both $\mu_5$ and $\mu_6$ do not commute with $\varphi.$ Indeed,
$\mu_5\varphi(1)=\gamma_5\gamma^*(y_1)=\gamma_5(y_1)=3$ and $\varphi\mu_5(1)=\varphi\gamma_5(1)=2$, and so $\mu_5\varphi(1)\neq \varphi\mu_5(1)$;
$\mu_6\varphi(1)=\gamma_6\gamma^*(y_1)=\gamma_6(y_1)=2$ and $\varphi\mu_6(1)=\varphi\gamma_6(1)=\varphi(y_3)=3,$ and so $\mu_6\varphi(1)\neq \varphi\mu_6(1)$.
As a consequence, we also have that  $C^{\langle\varphi\rangle}(p )\cap ST(p)=\varnothing$ and this completes the proof.
\end{proof}

\begin{proof}[Proof of Theorem \ref{gen-teo}] 
Suppose by contradiction that there exists a resolute, symmetric and Pareto optimal generalized matching mechanisms $\overline{F}$.
Let us denote by $\overline{\mathcal{P}}_*$ the set of the generalized preference profiles $\overline{p}$ such that, for every $x\in W$ and $y\in M$,
$y\succ_{\overline{p}(x)}x$ and $x\succ_{\overline{p}(y)}y$. Thus, we have that $\overline{p}\in \overline{\mathcal{P}}_*$ if, for every individual, being unmatched is the worst possible option. 

We claim that, for every $\overline{p}\in \overline{\mathcal{P}}_*$, $\overline{F}(\overline{p})\subseteq \mathcal{M}$. Indeed, assume by contradiction that there exists  $\overline{p}\in \overline{\mathcal{P}}_*$ and $\overline{\mu}\in \overline{\mathcal{M}}\setminus \mathcal{M}$ such that $\overline{F}(\overline{p})=\{\overline{\mu}\}$. Thus, there exists $z\in I$ such that $\overline{\mu}(z)=z$. If $\overline{\mu}=id_I$, then, considering any $x\in W$, $y\in M$ and 
$\overline{\mu}'=(xy)\in \overline{\mathcal{M}}$, we have that 
$\overline{\mu}'(x)=y\succ_{\overline{p}(x)}x=\overline{\mu}(x)$,
$\overline{\mu}'(y)=x\succ_{\overline{p}(y)}y=\overline{\mu}(y)$, and 
$\overline{\mu}'(z)=\overline{\mu}(z)$ for all $z\in I\setminus \{x,y\}$. 
In particular, $\overline{\mu}$ is not Pareto optimal for $\overline{p}$, and $\overline{F}$ is not Pareto optimal, a contradiction. If $\mu\neq id_I$, then the properties of $\overline{\mu}$ guarantees that $\overline{\mu}$ is product of 2-cycles whose support contains an element of $W$ and and element of $M$. As a consequence, the set $J=\{z\in I: \overline{\mu}(z)=z\}$ is nonempty and contains at least an element of $W$ and at least an element of $M$. Consider then  $x\in W\cap J$, $y\in M\cap J$ and $\overline{\mu}'=(xy)\overline{\mu}\in \overline{\mathcal{M}}$. We have that 
$\overline{\mu}'(x)=y\succ_{\overline{p}(x)}x=\overline{\mu}(x)$,
$\overline{\mu}'(y)=x\succ_{\overline{p}(y)}y=\overline{\mu}(y)$, and 
$\overline{\mu}'(z)=\overline{\mu}(z)$ for all $z\in I\setminus \{x,y\}$.
In particular, $\overline{\mu}$ is not Pareto optimal for $\overline{p}$, and $\overline{F}$ is not Pareto optimal, a contradiction.

Consider now $\Phi:\mathcal{P}\to \overline{\mathcal{P}}_*$ that associates with any $p\in\mathcal{P}$ the unique element $\Phi(p)$ of $\overline{\mathcal{P}}_*$ such that: 
for every $x\in W$ and $y_1,y_2\in M$, $y_1\succ_{\Phi(p)(x)} y_2$ if and only if $y_1\succ_{p(x)} y_2$;
for every $y\in M$ and $x_1,x_2\in W$, $x_1\succ_{\Phi(p)(y)} x_2$ if and only if $x_1\succ_{p(y)} x_2$;
for every $x\in W$ and $y\in M$, $y\succ_{\Phi(p)(x)}x$ and $x\succ_{\Phi(p)(y)}y$.
Having in mind the tables representing the elements of $\mathcal{P}$ and $\overline{\mathcal{P}}$, $\Phi(p)$ is built from the table associated with $p$ by simply adding at the bottom of each column the name of the corresponding individual. It can be easily shown that $\Phi$ is a bijection and that, for every $p\in \mathcal{P}$ and $\varphi\in G^*$, $\Phi(p)^\varphi=\Phi(p)^\varphi$.

Let us consider the matching mechanism $F$ defined, for every $p\in\mathcal{P}$, by $F(p)=\overline{F}(\Phi(p))$. It is a cumbersome exercise to show that $F$ is resolute, symmetric and Pareto optimal. Such a fact, by means of Proposition \ref{implications}, contradicts Theorem \ref{impossibilita1}.
\end{proof}

\section*{References}

\noindent Abdulkadiro{\u{g}}lu, A., S{\"o}nmez, T., 2003. School choice:
A mechanism design approach. American Economic Review 93, 729-747.
\vspace{2mm}

\noindent Bartholdi, L.,  Hann-Caruthers, W.,  Josyula, M., Tamuz, O., Yariv, L., 2021. Equitable Voting Rules. Econometrica 89, 563-589.
\vspace{2mm}

\noindent Bereczky, \`A.,  Mar\` oti, A., 2008. On groups with every normal subgroup transitive or semiregular. Journal of Algebra 319, 1733-1751.
\vspace{2mm}
 
\noindent Bubboloni, D., Gori, M., 2014. Anonymous and neutral majority
rules. Social Choice and Welfare 43, 377-401.
\vspace{2mm}

\noindent Bubboloni, D., Gori, M., 2015. Symmetric majority rules.
Mathematical Social Sciences 76, 73-86.
\vspace{2mm}

\noindent Bubboloni, D., Gori, M., 2016. Resolute refinements of social
choice correspondences. Mathematical Social Sciences 84,  37-49.
\vspace{2mm}

\noindent Bubboloni, D., Gori, M., 2021. Breaking ties in collective decision-making.
Decisions in Economics and Finance 44, 411-457.
\vspace{2mm}

\noindent Bubboloni, D., Gori, M., 2023. A generalization to networks of Young's characterization of the Borda rule. Available at arXiv:2211.06626v2.
\vspace{2mm}

\noindent Cooper, F., Manlove, D., 2020. Algorithms for new types of fair stable
matchings. In 18th International Symposium on Experimental Algorithms (SEA 2020). Leibniz International Proceedings in Informatics (LIPIcs), Volume 160,  pp. 20:1-20:13, Schloss Dagstuhl - Leibniz-Zentrum für Informatik (2020).
\vspace{2mm}

\noindent Do\u gan, O., Giritligil, A.E., 2015. Anonymous and neutral social choice: existence results on resoluteness. Murat
Sertel Center for Advanced Economic Studies, Working paper series 2015-01.
\vspace{2mm}

\noindent E\u gecio\u glu, \"{O}., 2009. Uniform generation of anonymous and neutral preference profiles for social choice rules.
Monte Carlo Methods and Applications 15, 241-255.
\vspace{2mm}

\noindent E\u gecio\u glu,  \"{O}., Giritligil, A.E., 2013. The impartial, anonymous, and neutral culture model: a probability model for sampling public preference structures. Journal of Mathematical Sociology 37, 203-222.
\vspace{2mm}

\noindent Endriss,  U.,  2020. Analysis of One-to-One Matching Mechanisms
via SAT Solving: Impossibilities for Universal Axioms.
The Thirty-Fourth AAAI Conference on Artificial Intelligence (AAAI-20), 1-8.
\vspace{2mm}

\noindent Gale, D., Shapley, L. S., 1962. College admissions and the
stability of marriage. The American Mathematical Monthly 69, 9-15.
\vspace{2mm}

\noindent Gusfield, D., Irving, R. W., 1989. {\it The stable marriage
problem: structure and algorithms}. The MIT Press.
\vspace{2mm}

\noindent Ha{\l}aburda, H., 2010. Unravelling in two-sided matching
markets and similarity of preferences. Games and Economic Behavior 69, 365-393.
\vspace{2mm}

\noindent Kivinen, S., 2023. On the manipulability of equitable voting rules. Games and Economic Behavior 141, 286-302.
\vspace{2mm}

\noindent  Kivinen, S., 2024. Equitable, neutral, and efficient voting rules. Journal of Mathematical Economics 115, 103061.
\vspace{2mm}

\noindent Klaus, B., Klijn, F., 2006. Procedurally fair and stable matching. Economic Theory 27, 431-447.
\vspace{2mm}

\noindent Knuth, D.E., 1976. {\it Mariages stables et leurs relations avec d'autres probl\`{e}mes combinatoires}. Montreal: Les Presses de l'Universit\'{e} de Montreal.
\vspace{2mm}

\noindent Kov\`acs, I., Malni\v{c}, A., Maru\v{s}i\v{c}, D., Miklavi\v{c}, \v{S}., 2015. Transitive group actions: (im)primitivity and semiregular groups. Journal of Algebraic Combinatorics 41, 867-885.
\vspace{2mm}

\noindent Maynard, P., Siemons, J., 2002. On the reconstruction index of the permutation groups: semiregular groups. Aequationes Mathematicae 64, 218-231.
\vspace{2mm}

\noindent Masarani, F., Gokturk, S.S., 1989. On the existence of fair matching algorithms.
Theory and Decision 26, 305-322.
\vspace{2mm}

\noindent Milne, J.S., 2021. {\it Group theory}. Available at www.jmilne.org/math/CourseNotes/GT.pdf 
\vspace{2mm}

\noindent Miyagawa, E., 2002. Strategy-proofness and the core in house allocation problems. Games and Economic Behavior
38, 347-361.
\vspace{2mm}

\noindent Nizamogullari, D., {\"O}zkal-Sanver, {\.I}., 2014. Characterization of the core in full domain marriage problems.
Mathematical Social Sciences 69, 34-42.

\noindent {\"O}zkal-Sanver, {\.I}., 2004. A note on gender fairness in
matching problems. Mathematical Social Sciences 47, 211-217.
\vspace{2mm}

\noindent Robinson, D. J. S, 1996. {\it A course in the theory of groups}. Graduate texts in Mathematics 80, Springer new York.
\vspace{2mm}

\noindent Romero-Medina, A., 2001. \lq Sex-equal\rq\  stable matchings.
Theory and Decision 150, 197-212.
\vspace{2mm}

\noindent Root, J., Bade, S., 2023. Royal processions: Incentives, efficiency and fairness in two-sided matching. In Proceedings of the 24th ACM Conference on Economics and Computation, p.1077. Full version of the paper at https://arxiv.org/abs/2301.13037
\vspace{2mm}

\noindent Roth, A. E., 1984. The evolution of the labor market for
medical interns and residents: a case study in game theory. Journal of
Political Economy 92, 991-1016.
\vspace{2mm}

\noindent Roth, A. E., 1991. A natural experiment in the organization of
entry-level labor markets: Regional markets for new physicians and
surgeons in the United Kingdom. The American Economic Review 81,
415-440.
\vspace{2mm}

\noindent Roth, A. E., Sotomayor, M. A., 1990. {\it Two-sided
matchings. A study in game-theoretic modeling and analysis}.
Cambridge University Press.
\vspace{2mm}

\noindent Sasaki, H., Toda, M., 1992. Consistency and characterization
of the core of two-sided matching problems. Journal of Economic Theory
56, 218-227.
\vspace{2mm}

\end{document}